\documentclass[10pt,twoside,twocolumn,journal]{IEEEtran}

\usepackage[utf8]{inputenc} 
\usepackage[T1]{fontenc}
\usepackage{url}
\usepackage{ifthen}
\usepackage{graphicx}
\usepackage[T1]{fontenc}

\usepackage{enumitem}
\usepackage[cmex10]{amsmath}
\usepackage{amssymb}

\usepackage{cite}
\usepackage{amsthm}
\usepackage{textcomp}
\usepackage{siunitx}
\usepackage{color}
\usepackage{cases}
\usepackage[font={small}]{caption}
\usepackage{epstopdf}
\usepackage{graphicx}
\usepackage{caption}
\usepackage{verbatim} 
\usepackage{bbm}
\usepackage{mathtools}

\DeclarePairedDelimiter\floor{\lfloor}{\rfloor}

\newtheorem{theorem}{Theorem}
\newtheorem{lemma}{Lemma}
\newtheorem{proposition}{Proposition}

\newtheorem*{problem}{Problem}

\theoremstyle{definition}
\newtheorem*{example}{Example}
\newtheorem{definition}{Definition}

\theoremstyle{remark}
\newtheorem{remark}{Remark}

\setlist[description]{style=multiline}

\begin{document}
\sloppy

\title{Communication-Aware Scheduling of Serial Tasks for Dispersed Computing }

\author{Chien-Sheng Yang,~\IEEEmembership{Student Member,~IEEE}, Ramtin Pedarsani,~\IEEEmembership{Member,~IEEE},\\ and A. Salman Avestimehr,~\IEEEmembership{Senior Member,~IEEE}
\thanks{This material is based upon work supported by Defense Advanced Research Projects Agency (DARPA) under Contract No. HR001117C0053, ARO award W911NF1810400, NSF grants CCF-1703575, CCF-1763673, NeTS-1419632, ONR Award No. N00014-16-1-2189, and the UC Office of President under grant No. LFR-18-548175. The views, opinions, and/or findings expressed are those of the author(s) and should not be interpreted as representing the official views or policies of the Department of Defense or the U.S. A part of this paper was presented in IEEE ISIT 2018 \cite{yang2018communication}.}
\thanks{C.-S.~Yang and A.~S.~Avestimehr are with the Department of Electrical and Computer Engineering, University of Southern California, Los Angeles, CA 90089 USA (e-mail: chienshy@usc.edu; avestimehr@ee.usc.edu).}%
\thanks{R.~Pedarsani is with the Department of Electrical and Computer
Engineering, University of California at Santa Barbara, Santa Barbara,
CA 93106, USA (e-mail: ramtin@ece.ucsb.edu).}}


\maketitle
\begin{abstract}
There is a growing interest in the development of in-network \emph{dispersed computing} paradigms that leverage the computing capabilities of heterogeneous resources dispersed across the network for processing massive amount of data collected at the edge of the network.
We consider the problem of task scheduling for such networks, in a dynamic setting in which arriving computation jobs are modeled as chains, with nodes representing tasks, and edges representing precedence constraints among tasks. In our proposed model, motivated by significant communication costs in dispersed computing environments, the communication times are taken into account. More specifically, we consider a network where servers can serve all task types, and sending the outputs of processed  tasks  from  one  server  to  another  server  results  in some communication delay. We first characterize the capacity region of the network, then propose a novel virtual queueing network encoding the state of the network. Finally, we propose a Max-Weight type scheduling policy, and considering the stochastic network in the fluid limit, we use a Lyapunov argument to show that the policy is throughput-optimal. 
Beyond the model of chains, we extend the scheduling problem to the model of directed acyclic graph (DAG) which imposes a new challenge, namely \textit{logic dependency} difficulty, requiring the data of processed parents tasks to be sent to the same server for processing the child task. We propose a virtual queueing network for DAG scheduling over broadcast networks, where servers always broadcast the data of processed tasks to other servers, and prove that Max-Weight policy is throughput-optimal.

\end{abstract}
\section{Introduction}\label{sec:intro}
\begin{figure*}[t]
  \centering
    \includegraphics[width = 0.55 \paperwidth]{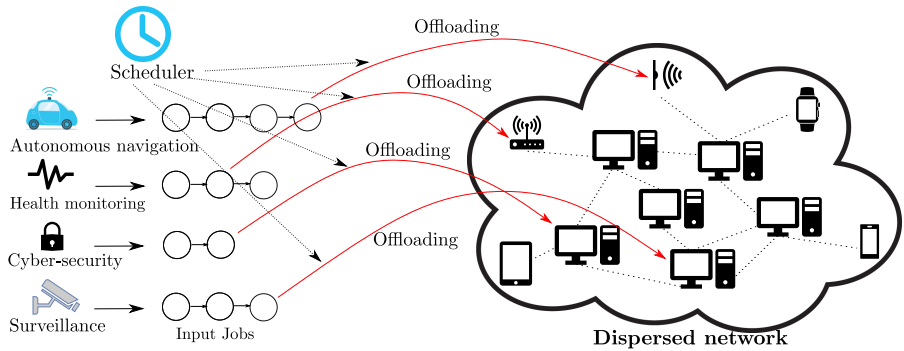}
\caption{Illustration of dispersed computing.}
\label{fig:dispersed_computing}
\end{figure*}
In many large-scale data analysis application domains, such as surveillance, autonomous navigation, and cyber-security, much of the needed data is collected at the edge of the network via a collection of sensors, mobile platforms, and users' devices. In these scenarios, continuous transfer of the massive amount of collected data from edge of the network to back-end servers (e.g., cloud) for processing incurs significant communication and latency costs. As a result, there is a growing  interest in development of in-network \emph{dispersed computing} paradigms that leverage the computing capabilities of heterogeneous resources dispersed across the network (e.g., edge computing, fog computing \cite{bonomi2012fog,chiang2016fog,hu2015mobile}). 

At a high level, a dispersed computing scenario (see Fig. \ref{fig:dispersed_computing}) consists of a group of networked computation nodes, such as wireless edge access points, network routers, and users' computers that can be utilized for offloading the computations.  There is, however, a broad range of computing capabilities that may be supported by different computation nodes. Some may perform certain kinds of operations at extremely high rate, such as high throughput matrix multiplication on GPUs, while the same node may perform worse on single threaded performance. Communication bandwidth between different nodes in dispersed computing scenarios can also be very limited and heterogeneous. Thus, for scheduling of computation tasks in such networks, it is critical to design efficient algorithms which carefully account for computation and communication heterogeneity.

In this paper, we consider the task scheduling problem in a dispersed computing network in which arriving jobs are modeled as chains, with nodes representing tasks, and edges representing precedence constraints among tasks. Each server is capable of serving all the task types and the service rate of a server depends on which task type it is serving.\footnote{The exponential distribution of servers' processing times is commonly observed in many computing scenarios (see e.g. \cite{lee2018speeding,reisizadeh2019coded}), and the considered geometric distribution in this paper is the equivalent of exponential distribution for discrete-time systems.} More specifically, after one task is processed by a server, the server can either process the children task locally or send the result to another server in the network to continue with processing of the children task. However, each server has a bandwidth constraint that determines the delay for sending the results. A significant challenge in this \textit{communication-aware} scheduling is that unlike traditional queueing networks, processed tasks are not sent from one queue to another queue probabilistically. Indeed, the scheduling decisions also determine the routing of tasks in the network. Therefore, it is not clear what is the maximum throughput (or, equivalently, the capacity region) that one can achieve in such networks, and what scheduling policy is throughput-optimal. This raises the following questions.
\begin{itemize}
    \item What is the capacity region of the network?
    \item What is a throughput-optimal scheduling policy for the network?
\end{itemize}

Our computation and network models are related to \cite{pedarsani2014scheduling,pedarsani2017robust}. However, the model that we consider in this paper is more general, as the communication times between servers are taken into account. In our network model, sending the outputs of processed tasks from one server to another server results in some communication constraints that make the design of efficient scheduling policy even more challenging. 

As the main contributions of the paper, we first characterize the capacity region of this problem (i.e., the set of all arrival rate vectors of computations for which there exists a scheduling policy that makes the network rate stable). To capture the complicated computing and communication procedures in the network, we propose a novel virtual queueing network encoding the state of the network. Then, we propose a Max-Weight type scheduling policy for the virtual queueing network, and show that it is throughput-optimal.

Since the proposed virtual queueing network is quite different from traditional queueing networks, it is not clear that the capacity region of the proposed virtual queueing network is equivalent to the capacity region of the original scheduling problem. Thus, to prove throughput-optimality Max-Weight policy, we first show the equivalence of two capacity regions: one for the dispersed computing problem that is characterized by a linear program (LP), and one for the virtual queueing network characterized by a mathematical optimization problem that is not an LP. Then, under the Max-Weight policy, we consider the stochastic network in the \emph{fluid limit}, and using a Lyapunov argument, we show that the fluid model of the virtual queueing network is \textit{weakly stable} \cite{dai1995positive} for all arrival vectors in the capacity region, and \textit{stable} for all arrival vectors in the interior of the capacity region. This implies that the Max-Weight policy is throughput-optimal for the virtual queueing network as well as for the original scheduling problem.       

Finally, we extend the scheduling problem for dispersed computing to a more general computing model, where jobs are modeled as directed acyclic graphs (DAG). Modeling a job as a DAG incurs more complex \textit{logic dependencies} among the smaller tasks of the job compared to chains. More precisely, the logic dependency difficulty arises due to the requirement that the data of processed parents tasks have to be sent to the \emph{same} server for processing child tasks. To resolve this logic dependency difficulty, we consider a specific class of networks, named \textit{broadcast network}, where servers in the network always broadcast the data of processed tasks to other servers, and propose a virtual queueing network for the DAG scheduling problem. We further demonstrate that Max-Weight policy is throughput-optimal for the proposed virtual queueing network.

In the following, we summarize the key contributions in this paper:
\begin{itemize}
    \item We characterize the capacity region for the new network model. 
    \item To capture the heterogeneity of computation and communication in the network, we propose a novel virtual queueing network. 
    \item We propose a Max-Weight type scheduling policy, which is throughput-optimal for the proposed virtual queueing network.
    \item For the communication-aware DAG scheduling problem for dispersed computing, we demonstrate that Max-Weight policy is throughput-optimal for broadcast networks. 
\end{itemize}
\vspace{0.3cm}
\textbf{Related Work:} Task scheduling problem has been widely studied in the literature, which can be divided into two main categories: static scheduling and dynamic scheduling. In the static or offline scheduling problem, jobs are present at the beginning, and the goal is to allocate tasks to servers such that a performance metric such as average computation delay is minimized. In most cases, the static scheduling problem is computationally hard, and various heuristics, approximation and stochastic approaches are proposed (see e.g. \cite{kwok1999static,zheng2013stochastic,tang2011stochastic,li2015scheduling,chen2009ant,blythe2005task,yu2006scheduling,topcuoglu2002performance}). Given a task graph over heterogeneous processors, \cite{topcuoglu2002performance} proposes Heterogeneous Earliest Finish Time (HEFT) algorithm which first prioritizes tasks based on the dependencies in the graph, and then assign tasks to processors starting with the highest priority. In the scenarios of edge computing, static scheduling problem has been widely investigated in recent years \cite{kao2017hermes,jia2014heuristic,mahmoodi2016optimal,zhang2015collaborative}. To minimize computation latency while meeting prescribed constraints, \cite{kao2017hermes} proposes a polynomial time approximation scheme algorithm with theoretical performance guarantees. \cite{jia2014heuristic} proposes an heuristic online task offloading algorithm which makes the parallelism between the mobile and the cloud maximized by using a load-balancing approach. \cite{mahmoodi2016optimal} proposes an optimal wireless-aware joint scheduling and computation offloading scheme for multicomponent applications. In \cite{zhang2015collaborative}, under a stochastic wireless channel, collaborative task execution between a mobile device and a cloud clone for mobile applications has been investigated.

In the online scheduling problem, jobs arrive to the network according to a stochastic process, and get scheduled dynamically over time. In many works in the literature, the tasks have dedicated servers for processing, and the goal is to establish stability conditions for the network \cite{baccelli1989acyclic,baccelli1989fork}. Given the stability results, the next natural goal is to compute the expected completion times of jobs or delay distributions. However, few analytical results are available for characterizing the delay performance, except for the simplest models. One approach to understand the delay performance of stochastic networks is analyzing the network in ``heavy-traffic'' regime. See for example \cite{varma1990heavy,nguyen1993processing,stolyar2004maxweight}. 
When the tasks do not have dedicated servers, one aims to find a throughput-optimal scheduling policy (see e.g. \cite{eryilmaz2005stable}), i.e. a policy that stabilizes the network, whenever it can be stabilized. Max-Weight scheduling, proposed in \cite{tassiulas1992stability,dai2005maximum}, is known to be throughput-optimal for wireless networks, flexible queueing networks \cite{neely2005dynamic,eryilmaz2007fair,walton2014concave} and data centers networks \cite{maguluri2012stochastic}. In \cite{feng2017approximation,zhang2018optimal}, the chain-type computation model is also considered for distributed computing networks. However, our network model is more general as it captures the computation heterogeneity in dispersed computing networks, e.g., the service rate of a server in our network model depends on which task type it serves.


\textbf{Notation.} We denote by $[N]$ the set $\{1, \dots , N\}$ for any positive integer $N$. For any two vectors $\vec{x}$ and $\vec{y}$, the notation $\vec{x} \leq \vec{y}$ means that $x_{i} \leq y_{i}$ for all $i$.
\section{System Model}\label{sec:sys}
\begin{figure*}[t]
  \centering
    \includegraphics[width = 0.55\paperwidth]{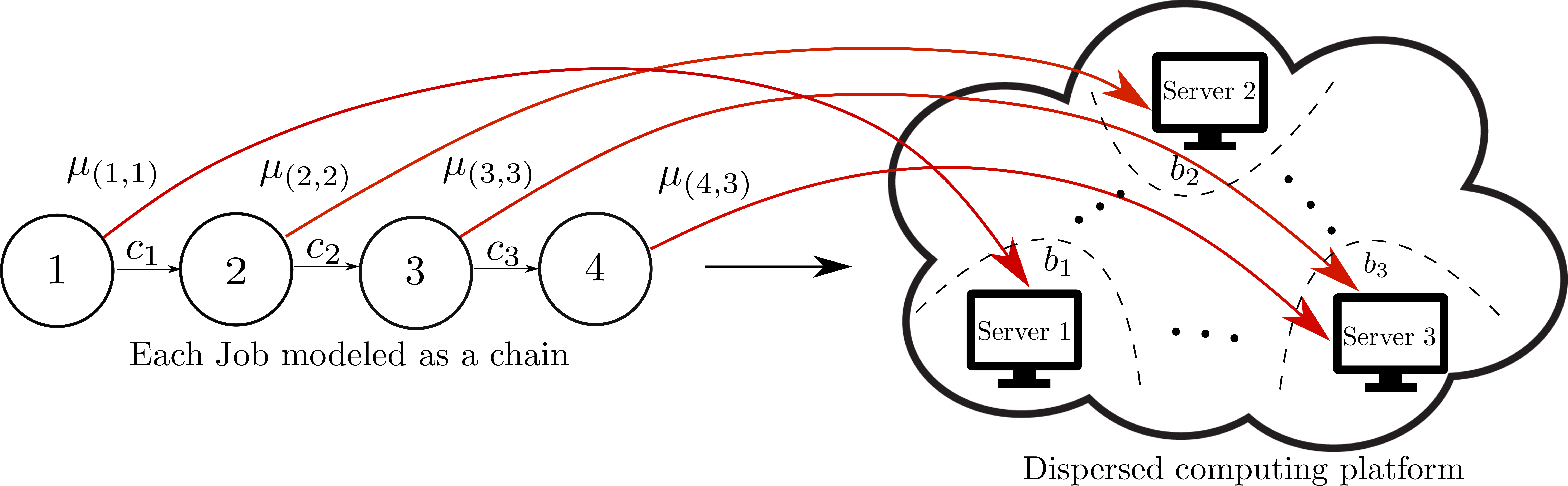}
\caption{A simple example of task scheduling for dispersed computing.}
\label{fig:chain}
\end{figure*}
\subsection{Computation Model}\label{subsec:computation_model}
As shown in Fig. \ref{fig:chain}, each job is modeled as a chain which includes serial tasks. Each node of the chain represents one task type, and each (directed) edge of the chain represents a precedence constraint. Moreover, we consider $M$ types of jobs, where each type is specified by one chain structure.

For this problem, we define the following terms. Let $(\mathcal{I}_m,\{c_k\}_{k \in \mathcal{I}_m})$ be the chain corresponding to the job of type $m$, $m \in [M]$, where $\mathcal{I}_m$ denotes the set of nodes of type-$m$ jobs, and $c_k$ denotes the data size (bits) of output type-$k$ task. Let the number of tasks of a type-$m$ job be $K_m$, i.e. $|\mathcal{I}_m| = K_m$, and the total number of task types in the network be $K$, so that $\sum^{M}_{m=1} K_m = K$. We assume that jobs are independent with each other which implies $\mathcal{I}_1,\mathcal{I}_2,\dots,\mathcal{I}_m$ are disjoint. Thus, we can index the task types in the network by $k$, $k \in [K]$, starting from job type $1$ to $M$. Therefore, task type-$k$ belongs to job type $m(k)$ if
\begin{align*}
    \sum^{m(k)-1}_{m^{'}=1} K_{m^{'}} < k \leq \sum^{m(k)}_{m^{'}=1} K_{m^{'}}. 
\end{align*}
We call task $k^{'}$ a parent of task $k$ if they belong to the same chain and there is a directed edge from $k^{'}$ to $k$. Without loss of generality, we let task $k$ be the parent of task $k+1$, if task $k$ and task $k+1$ belong to the same chain, i.e. $m(k)=m(k+1)$. In order to process task $k+1$, the processing of task $k$ should be completed. Node $k$ is said to be the root of chain type $m(k)$ if $k =1+ \sum^{m(k)-1}_{m^{'}=1}K_{m^{'}}$. We denote $\mathcal{C}$ as the set of the root nodes of the chains, i.e. $\mathcal{C} = \{k: k =1+ \sum^{i-1}_{m=1}K_{m}, \ \forall \ i \in [M]\}$. Also, node $k$ is said to be the last node of chain type $m(k)$ if $k = \sum^{m(k)}_{m^{'}=1}K_{m^{'}}$. Then, we denote $\mathcal{H}$ as the set of the last nodes of the chains, i.e. $\mathcal{H} = \{k: k = \sum^i_{m=1}K_m, \ \forall \  i \in [M]\}$.
\subsection{Network Model}\label{subsec:network_model}
In the dispersed computing network, as shown in Fig. \ref{fig:chain}, there are $J$ servers which are connected with each other. Each server can serve all types of tasks. We consider the network in discrete time. We assume that the arrival process of jobs of type $m$ is a Bernoulli process with rate $\lambda_m$, $0 < \lambda_m < 1$; that is, in each time slot a job of type $m$ arrives to the network with probability $\lambda_m$ independently over time. We assume that the service times for the nodes are geometrically distributed, independent across time slots and across different nodes, and also independent from the arrival processes. When server $j$ processes type-$k$ tasks, the service completion time has mean $\mu^{-1}_{(k,j)}$. Thus, $\mu_{(k,j)}$ can be interpreted as the service rate of type-$k$ task when processed by server $j$. Similarly, we model the communication times between two servers as geometric distribution,  which are independent across time slots and across different nodes, and also independent from the arrival processes. When server $j$ communicates data of size $1$ bit to another server, the communication time has mean $b^{-1}_j$. Therefore, $b_j$ can be interpreted as the average bandwidth (bits/time slot) of server $j$ for communicating data of processed tasks. Without loss of generality, the system parameters can always be rescaled so that $\frac{b_j}{c_k}<1$ for all $k$ and $j$, by speeding up the clock of the system. We assume that each server is able to communicate data and process tasks at the same time slot. In the task scheduling problem of dispersed computing, tasks are scheduled on servers based on a scheduling policy. After a task is served by a server, a scheduling policy is to determine where the data of processed task should be sent to for processing the child task.
\subsection{Problem Formulation}
Given the above computation model and network model, we formulate the task scheduling problem of dispersed computing based on the following terms. 
\begin{definition}
Let $Q^n$ be a stochastic process of the number of jobs in the network over time $n \geq 0$. A network is \textit{rate stable} if
\begin{align}
\lim_{n \rightarrow \infty} \frac{Q^n}{n} = 0 \quad \text{almost surely}.
\end{align}
\end{definition}
\begin{definition}[Capacity Region]
\label{def:capacity}
We define the \textit{capacity region} of the network to be the set of all arrival rate vectors where there exists a scheduling policy that makes the network rate stable. 
\end{definition}
\begin{definition}
The fluid level of a stochastic process $Q^n$, denoted by $X(t)$, is defined as
\begin{equation}
    X(t) = \lim_{r \to \infty}  \frac{Q^{\floor{rt}}}{r}.
\end{equation}
\end{definition}
\begin{definition}
Let $X(t)$ be the fluid level of a stochastic process. The fluid model of the the process is \textit{weakly stable}, if $X(0)=0$ for $t=0$, then $X(t)=0$ for all $t \geq 0$. \cite{dai1995positive}
\end{definition}
Note that we later model the network as a network of virtual queues. Since the arrival and service processes are memoryless, given a scheduling policy, the queue-length vector in this virtual queueing network is a Markov process. 
\begin{definition}
A network is \textit{strongly stable} if its underlying Markov process is positive recurrent for all the arrival rate vectors in the interior of the capacity region. 
\end{definition}
\begin{definition}
\label{def:throughput_optimal}
A scheduling policy is \textit{throughput-optimal} if, under this policy, the network is rate stable for all arrival rate vectors in the capacity region; and strongly stable for all arrival rate vectors in the interior of the capacity region.  
\end{definition}
Based on above definitions, our problem is now formulated as the following.
\begin{problem}
Consider a dispersed computing network consisting of network and computation models as defined in Sections \ref{subsec:computation_model} and \ref{subsec:network_model}, we pose the following two questions:
\begin{itemize}
    \item What is the capacity region of the network as defined in Definition \ref{def:capacity}?
    \item What is a throughput-optimal scheduling policy for the network as defined in Definition \ref{def:throughput_optimal}?
\end{itemize}
\end{problem}

\section{Capacity Region Characterization}\label{sec:cap}
As mentioned previously, our goal is to find a throughput-optimal scheduling policy. Before that, we characterize the capacity region of the network. 

Now, we consider an arbitrary scheduling policy and define two allocation vectors to characterize the scheduling policy. Let $p_{(k,j)}$ be the long-run fraction of capacity that server $j$ allocates for processing available type-$k$ tasks. We define $\vec{p}$ to be the capacity allocation vector. An allocation vector $\vec{p}$ is \textit{feasible} if 
\begin{equation} \label{eq:p}
    \sum^{K}_{k=1}p_{(k,j)} \leq 1, \; \forall \; j \in [J].
\end{equation}

Let $q_{(k,j)}$ be the long-run fraction of the bandwidth that server $j$ allocates for communicating data of processed type-$k$ tasks. We can define $\vec{q}$ to be the bandwidth allocation vector. Therefore, an allocation vector $\vec{q}$ is feasible if
\begin{equation} \label{eq:q}
    \sum_{k \in [K] \backslash \mathcal{H}} q_{(k,j)} \leq 1, \; \forall \; j \in [J].
\end{equation}

Given a capacity allocation vector $\vec{p}$, consider task $k$ and task $k+1$ which are in the same chain on server $j$. As time $t$ is large, up to time $t$, the number of type-$k$ tasks processed by server $j$ is $\mu_{(k,j)}p_{(k,j)}t$ and the number of type-($k+1$) tasks processed by server $j$ is $\mu_{(k+1,j)}p_{(k+1,j)}t$. Therefore, as $t$ is large, up to time $t$, the number of type-($k+1$) tasks that server $j$ is not able to serve is
\begin{equation}
    \mu_{(k,j)}p_{(k,j)}t-\mu_{(k+1,j)}p_{(k+1,j)}t
\end{equation}
Clearly, the type-($k+1$) tasks which cannot be served by server $j$ have to be processed by other servers. Hence, up to time $t$ and $t$ is large, server $j$ has to at least communicate data of $\mu_{(k,j)}p_{(k,j)}t-\mu_{(k+1,j)}p_{(k+1,j)}t$ processed type-$k$ tasks to other servers.

On the other hand, given a bandwidth allocation vector $\vec{q}$, up to time $t$ and $t$ is large, the number of the type-$k$ tasks communicated by server $j$ is $\frac{b_jq_{(k,j)}t}{c_k}$. Therefore, to make the network stable, we obtain the following constraints:
\begin{equation}
    \frac{b_jq_{(k,j)}}{c_k} \geq \mu_{(k,j)}p_{(k,j)} - \mu_{(k+1,j)}p_{(k+1,j)}
\end{equation}
$\forall \ j \in [J]$ and $\forall \ k \in [K] \backslash \mathcal{H}$.

For this scheduling problem, we can define a linear program (LP) that characterizes the capacity region of the network, defined to be the rate vectors $\vec{\lambda}$ for which there is a scheduling policy with corresponding allocation vectors $\vec{p}$ and $\vec{q}$ such that the network is rate stable. The nominal traffic rate to all nodes of job type $m$ in the network is $\lambda_m$. Let $\nu_k(\vec{\lambda})$ be the nominal traffic rate to the node of task $k$ in the network. Then, $\nu_k(\vec{\lambda})=\lambda_m$ if $m(k)=m$. The LP that characterizes capacity region of the network makes sure that the total service capacity allocated to each node in the network is at least as large as the nominal traffic rate to that node, and the communication rate of each server is at least as large as the rate of task that the server is not capable of serving. Then, the LP known as the \textit{static planning problem (SPP)} \cite{harrison2000brownian} - is defined as follows:

\textbf{Static Planning Problem (SPP):} 
\begin{align}
\quad \text{Maximize} \quad  & \delta\\
    \text{subject to} \quad & \nu_k(\vec{\lambda}) \leq \sum^{J}_{j=1}\mu_{(k,j)}p_{(k,j)} - \delta, \ \forall \ k  \label{eq:LP}\\
     & \frac{b_jq_{(k,j)}}{c_k}-\delta \geq \mu_{(k,j)}p_{(k,j)} - \mu_{(k+1,j)}p_{(k+1,j)}, \nonumber \\
     & \ \forall \ j, \ \forall \ k \in [K] \backslash \mathcal{H}\\
    & 1 \geq \sum^{K}_{k=1}p_{(k,j)}, \ \forall \ j\\
    & 1 \geq \sum_{k \in [K] \backslash \mathcal{H}} q_{(k,j)}, \ \forall \ j\\
    & \vec{p} \geq \vec{0},\ \vec{q} \geq \vec{0}.
\end{align}
Based on SPP above, the capacity region of the network can be characterized by following proposition.
\begin{proposition}\label{optimal_LP}
The \textit{capacity region} $\Lambda$ of the network characterizes the set of all rate vectors $\vec{\lambda} \in \mathbb{R}^M_{+}$ for which the corresponding optimal solution $\delta^*$ to the static planning problem (SPP) satisfies $\delta^* \geq 0$. In other words, \textit{capacity region} $\Lambda$ of the network is characterized as follows 
\begin{align*}
   & \Lambda \triangleq \Bigg\{\vec{\lambda} \in \mathbb{R}^{M}_{+}: \exists \ \vec{p} \geq \vec{0}  , \ \vec{q} \geq \vec{0} \ \text{s.t.} \ \sum^{K}_{k=1}p_{(k,j)} \leq 1 \ \forall \ j,\\
   & \sum_{k \in [K] \backslash \mathcal{H}}q_{(k,j)} \leq 1 \ \forall \ j,
   \ \nu_k(\vec{\lambda}) \leq \sum^{J}_{j=1}\mu_{(k,j)}p_{(k,j)} \ \forall \ k, \ \text{and}\\ &\frac{b_jq_{(k,j)}}{c_k} \geq \mu_{(k,j)}p_{(k,j)} - \mu_{(k+1,j)}p_{(k+1,j)} \ \forall \ j, \ \forall \ k \in [K] \backslash \mathcal{H}\Bigg\}.
  \end{align*}
 \end{proposition}
\begin{proof}
We show that $\delta^{*} \geq 0$ is a necessary and sufficient condition for the rate stability of the network. Consider the network in the fluid limit (See \cite{dai1995positive} for more details on the stability of fluid models). At time $t$, we denote fluid level of type-$k$ tasks in the network as $X_k(t)$,  fluid level of type-$k$ tasks served by server $j$ as $X_{(k,j)}(t)$ and fluid level of type-$k$ tasks sent by server $j$ as $X_{(k,j),c}(t)$.

The dynamics of the fluid are as follows
\begin{equation}
    X_k(t) = X_k(0) +\lambda_{m(k)} t -D_k(t)
\end{equation}
where $\lambda_{m(k)} t$ is the total number of jobs of type $m$ that
have arrived to the network until time $t$, and $D_k(t)$ is the
total number of type-$k$ tasks that have been processed up to
time $t$ in the fluid limit. For $\forall \ k \in [K] \backslash \mathcal{H}$, because of flow conservation, we have
\begin{equation}\label{eq:proposition}
    X_{(k,j)}(t) - X_{(k,j),c}(t) \leq X_{(k+1,j)}(t).
\end{equation}

Suppose $\delta^* <0$. Let's show that the network is weakly unstable, i.e., if $X_k(0) = 0$ for all $k$, there exists $t_0$ and $k$ such that $X_k(t_0) > 0$. In contrary, suppose that there exists a scheduling policy such that under that
policy for all $t \geq 0$ and all $k$, $X_k(t)=0$. Now, we pick a regular point $t_1$ which means $X_k(t_1)$ is differentiable at $t_1$ for all $k$. Then, for all $k$, $\dot{X}_k(t_1) = 0$ which implies that $\dot{D}_k(t_1) = \lambda_{m(k)}=\nu_k(\lambda)$. At a regular point $t_1$, $\dot{D}_k(t_1)$ is exactly the total service capacity allocated to type-$k$ tasks at time $t_1$. This implies that there exists $p_{(k,j)}$ at time $t_1$ such that $\nu_k(\lambda)=\sum^J_{j=1}\mu_{(k,j)}p_{(k,j)}$ for all $k$. Furthermore, from (\ref{eq:proposition}), we have 
\begin{equation}
    \mu_{(k,j)}p_{(k,j)}t_1 - \frac{b_jq_{(k,j)}}{c_k}t_1 \leq  \mu_{(k+1,j)}p_{(k+1,j)}t_1
\end{equation}
which implies that there exists $q_{(k,j)}$ at time $t_1$ such that
\begin{equation}
        \frac{b_jq_{(k,j)}}{c_k} \geq \mu_{(k,j)}p_{(k,j)}-\mu_{(k+1,j)}p_{(k+1,j)}
\end{equation}
$\forall \ k \in [K] \backslash \mathcal{H}$. However, this contradicts $\delta^* < 0$.

Now suppose that $\delta^{*} \geq 0$, $\vec{p}^{\,*}$ and $\vec{q}^{\,*}$ are the capacity allocation vector and bandwidth allocation vector respectively that solve SPP. Now, let us consider a generalized head-of-the-line processor sharing policy that server $j$ works on type-$k$ tasks with capacity $p^*_{(k,j)}$ and communicates the processed data of type-$k$ tasks with bandwidth $b_jq^*_{(k,j)}$. Then the cumulative service allocated to type-$k$ tasks up to time $t$ is $\sum^J_{j=1}\mu_{(k,j)}p^*_{(k,j)}t \geq (\nu_k(\lambda) + \delta^*)t$. Thus, we have $\dot{X}_k(t)= \nu_k(\lambda)-\sum^J_{j=1}\mu_{(k,j)}p^*_{(k,j)} \leq - \delta^* \leq 0$ for all $t > 0$ and all $k$. If $X_k(0)=0$ for all $k$, then $X_k(t)=0$ for all $t\geq 0$ and all $k$ which implies that the network is weakly stable, i.e., the network is rate stable \cite{dai1995positive}.
\end{proof}
\begin{example}
 We consider that there are two types of jobs arriving to a network, in which $K_1=2$, $K_2=3$ and $c_k=1$, $\forall \ k$. There are $2$ servers in the network, where the service rates are: $\mu_{(1,1)}=4$, $\mu_{(1,2)}=3$, $\mu_{(2,1)}=2$, $\mu_{(2,2)}=4$, $\mu_{(3,1)}=2.5$, $\mu_{(3,2)}=3.5$, $\mu_{(4,1)}=0.5$, $\mu_{(4,2)}=4.5$, $\mu_{(5,1)}=3.5$, $\mu_{(5,2)}=1$ and average bandwidths are $b_1=1.5$, $b_2=1$. As shown in Fig. \ref{fig:capacity}, the capacity region of the previous model without communication constraints \cite{pedarsani2014scheduling,pedarsani2017robust} is larger than the capacity region of our proposed model with communication constraints.
\end{example}
\begin{figure}[t]
    \centering
    \includegraphics[width=\columnwidth]{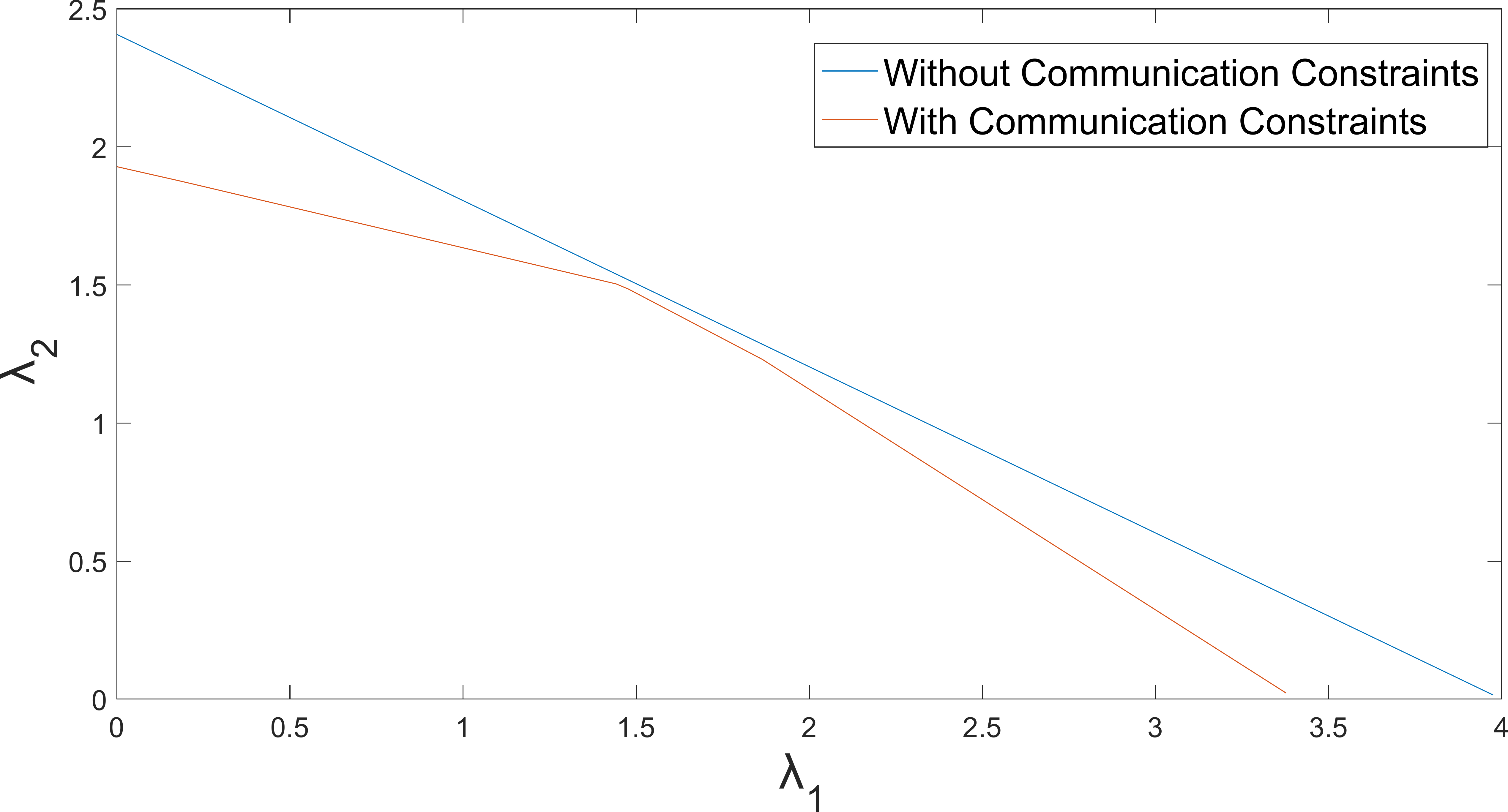}
    \caption{The comparison of capacity regions between the previous model without communication constraints \cite{pedarsani2014scheduling,pedarsani2017robust} and the proposed model with communication constraints in this paper.}\label{fig:capacity}
\end{figure}

\section{Queueing Network Model}\label{sec:queue}
\begin{figure*}[!t]
\minipage{0.32\textwidth}
  \includegraphics[width=\linewidth]{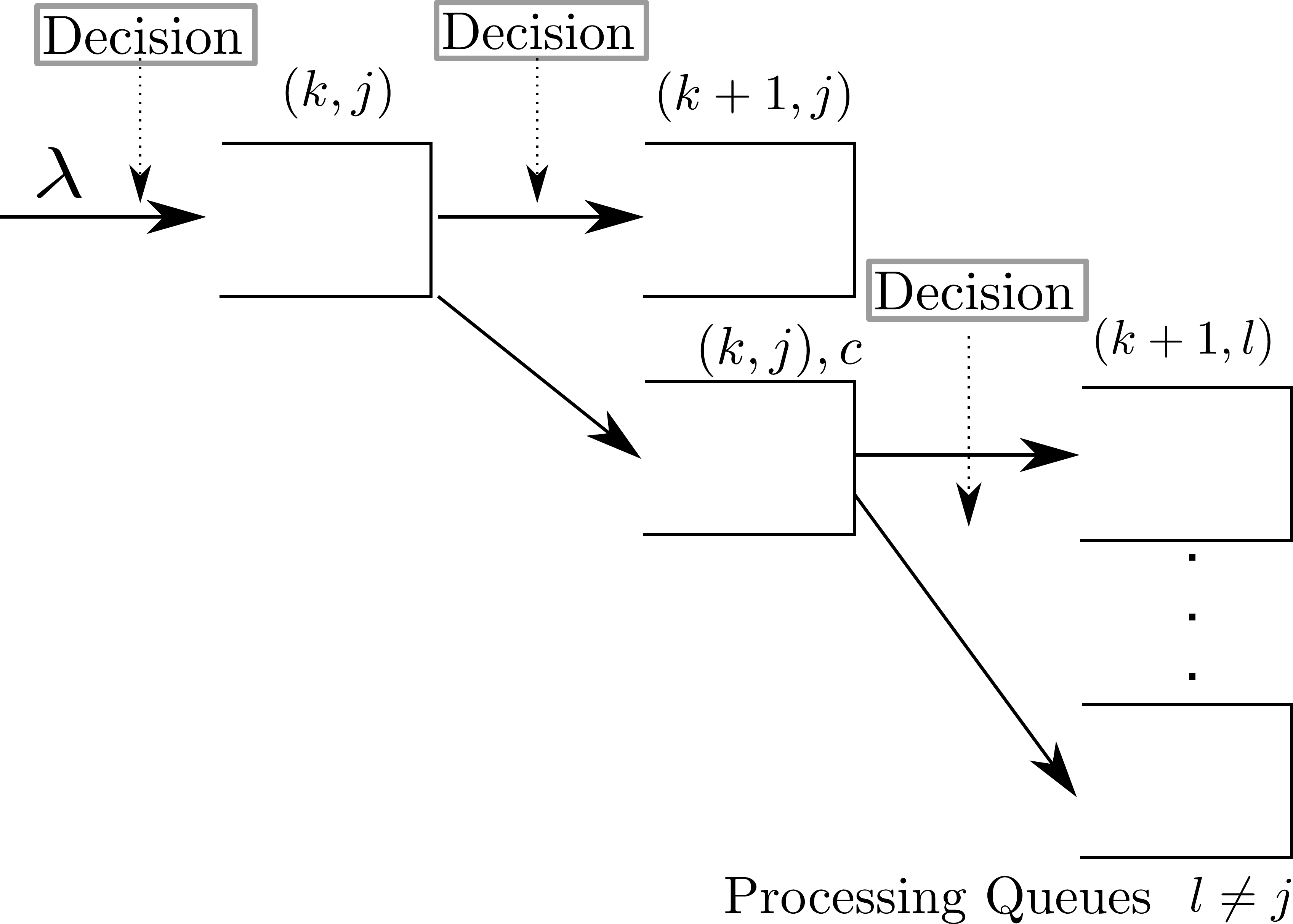}
  \caption{$k$ is a root of one chain ($k \in \mathcal{C}$).}\label{fig:queue_root}
\endminipage\hfill
\minipage{0.37\textwidth}
  \includegraphics[width=\linewidth]{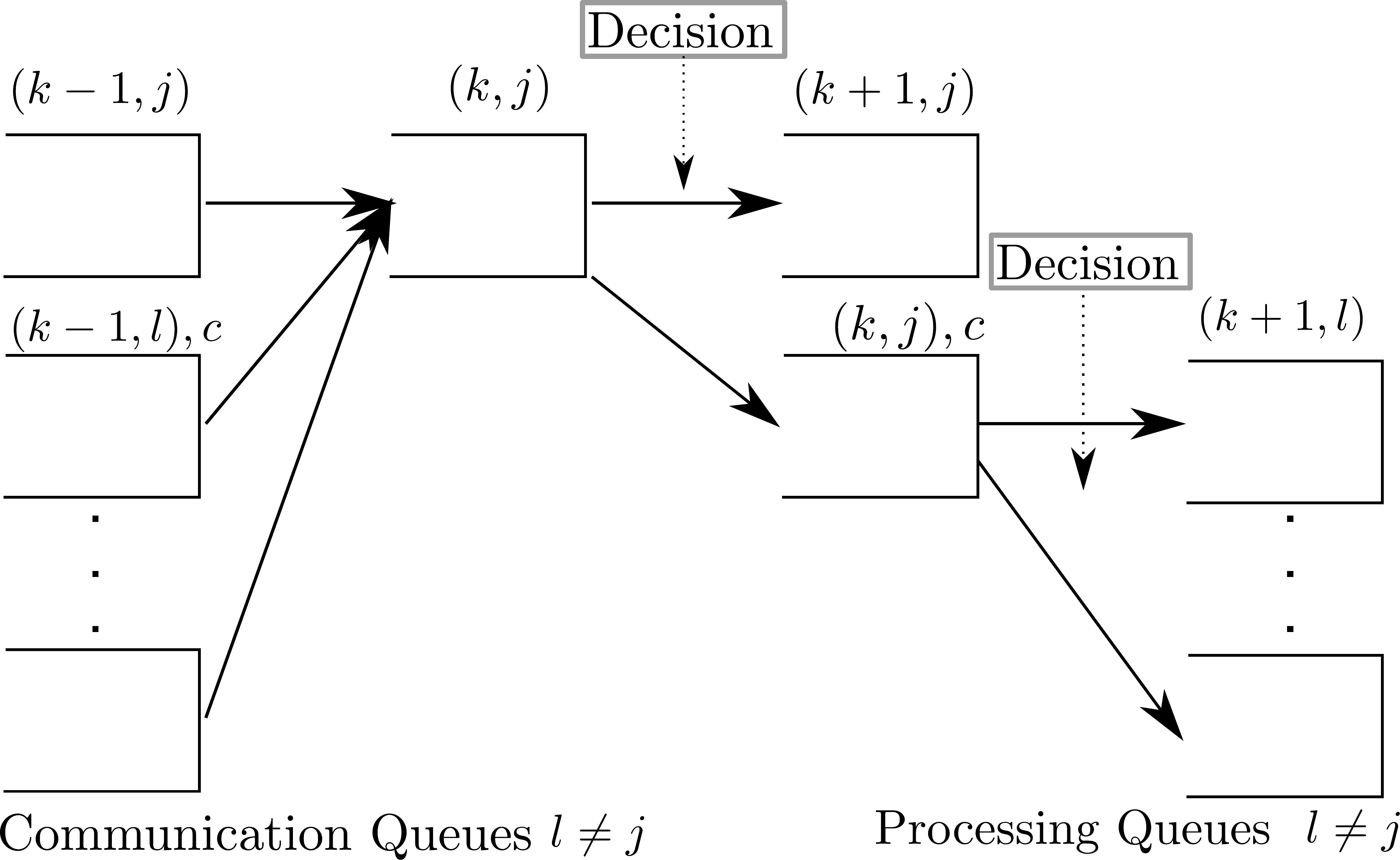}
  \caption{$k$ is not a root of one chain ($k \notin \mathcal{C}$).}\label{fig:queue_notroot}
\endminipage\hfill
\minipage{0.18\textwidth}%
  \includegraphics[width=\linewidth]{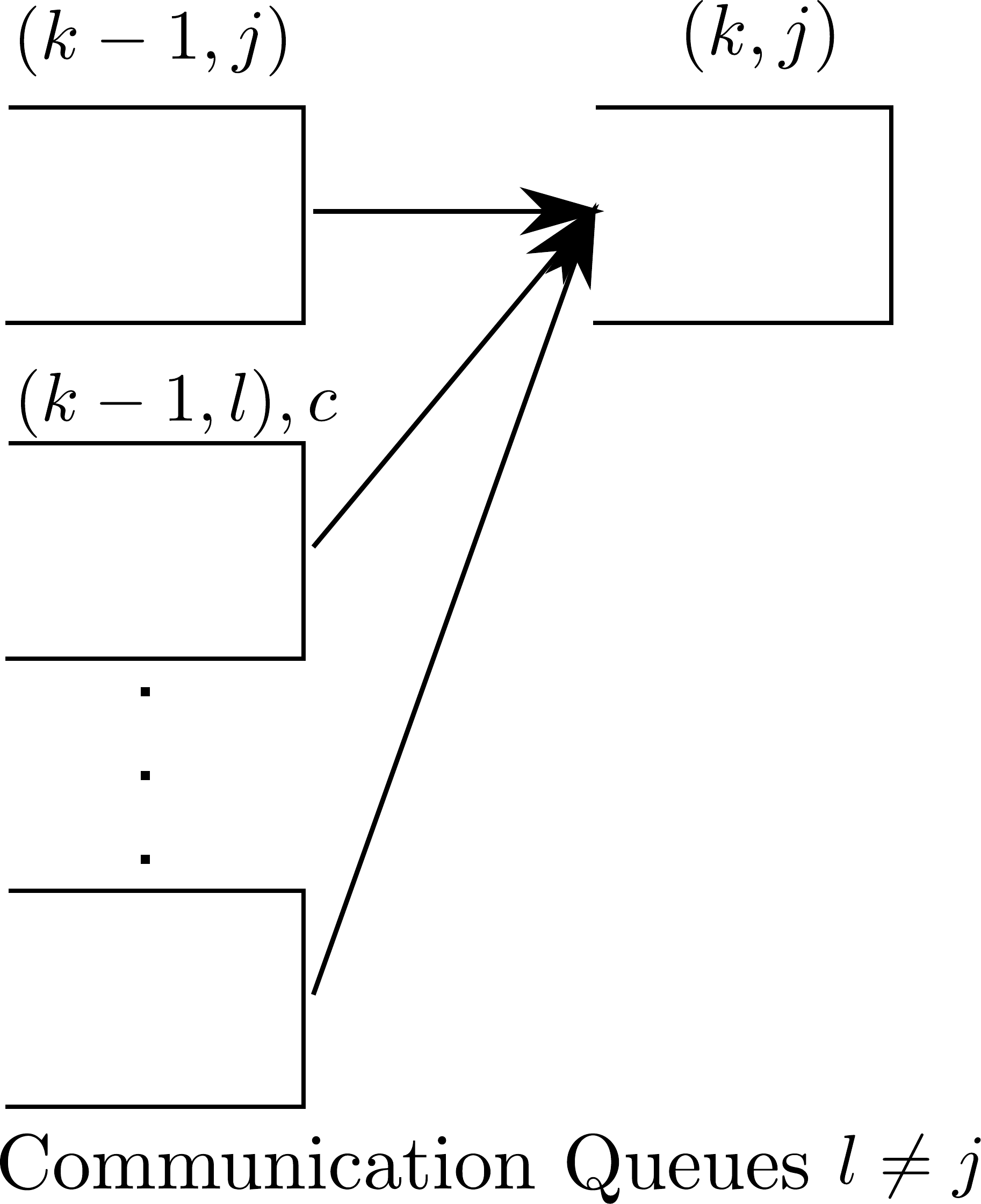}
  \caption{$k \in \mathcal{H}$.}\label{fig:queue_lastnode}
\endminipage
\end{figure*}
In this section, in order to find a throughput-optimal scheduling policy, we first design a virtual queueing network that encodes the state of the network. Then, we introduce an optimization problem called \textit{queueing network planning problem} for the virtual queueing network to characterize the capacity region of this virtual queueing network.
\subsection{Queueing Network}
Based on the computation model and network model described in Section \ref{sec:sys}, let's illustrate how we model a queueing network. The queueing network consists of two kinds of queues, \textit{processing queue} and \textit{communication queue}, which are modeled in the following manner:
\begin{enumerate}
    \item \textbf{Processing Queue}: We maintain one virtual queue called $(k,j)$ for type-$k$ tasks which are processed at server $j$.
    \item \textbf{Communication Queue}: For $k \notin \mathcal{H}$, we maintain one virtual queue called $(k,j),c$ for processed type-$k$ tasks to be sent to other servers by server $j$. 
\end{enumerate}
Therefore, there are $(2K-M)J$ virtual queues in the queueing network. Concretely, the queueing network can be shown as illustrated in Fig. \ref{fig:queue_root}, Fig. \ref{fig:queue_notroot} and Fig. \ref{fig:queue_lastnode}. 

Now, we describe the dynamics of the virtual queues in the network. Let's consider one type of job which consists of serial tasks. As shown in Fig. \ref{fig:queue_root}, a root task $k$ of the job is sent to processing queue $(k,j)$ if the task $k$ is scheduled on server $j$ when a new job comes to the network. For any node $k$ in this chain, the result in processing queue $(k,j)$ is sent to processing queue $(k+1,j)$ if task $k+1$ is scheduled on server $j$. Otherwise, the result is sent to communication queue $(k,j),c$. If task $k+1$ in queue $(k,j),c$ is scheduled on server $l$, it is sent to queue $(k+1,l)$, where $l \in [J] \backslash \{j\}$.  As shown in Fig. \ref{fig:queue_root}, if $k$ is a root of one chain, the traffic to processing queue $(k,j)$ is only the traffic of type-$m(k)$ jobs coming to the network. Otherwise, as shown in Fig. \ref{fig:queue_notroot}, the traffic to processing queue $(k,j)$ is from processing queue $(k-1,j)$ and communication queues $(k-1,l),c$, $\forall \ l \in [J] \backslash \{j\}$. Furthermore, the traffic to communication queue $(k,j),c$ is only from processing queue $(k,j)$, where $k \in [K] \backslash \mathcal{H}$.

Let $Q_{(k,j)}$ denote the length of processing queue $(k,j)$ and $Q_{(k,j),c}$ denote the length of communication queue $(k,j),c$. A task of type $k$ can be processed by server $j$ if and only if $Q_{(k,j)} > 0$ and a processed task of type $k$ can be sent by server $j$ to other servers if and only if $Q_{(k,j),c} > 0$. Let $d^n_{(k,j)} \in \{0,1\}$ be the number of processed tasks of type $k$ by server $j$ at time $n$, $a^n_{m(k)} \in \{0,1\}$ be the number of jobs of type $m$ that arrives to the network at time $n$ and $d^n_{(k,j),c} \in \{0,1\}$ be the number of processed type-$k$ tasks sent to other servers by server $j$ at time $n$. We denote $u^n_{m \rightarrow j} \in \{0,1\}$ as the decision variable that root task of the type-$m$ job is scheduled on server $j$ at time $n$, i.e.,
\begin{align*}
    u^n_{m \rightarrow j} =
     \begin{cases}
     1 & \text{root task of job $m$ is scheduled on server $j$}\\ & \text{at time $n$}\\
     0 & \text{otherwise}.\\
     \end{cases}
\end{align*}
We denote $w^n_{k,j \rightarrow l}\in \{0,1\}$ as the decision variable that processed type-$k$ task in $(k,j),c$ is sent to $(k+1,l)$ at time $n$, i.e., 
\begin{align*}
    w^n_{k,j \rightarrow l} =
     \begin{cases}
     1 & \text{processed type-$k$ task in $(k,j),c$ is sent to}\\ &\text{$(k+1,l)$ at time $n$}\\
     0 & \text{otherwise}.\\
     \end{cases}
\end{align*}
Moreover, let $s^n_{k,j \rightarrow j}\in \{0,1\}$ as the decision variable that processed type-$k$ task in queue $(k,j)$ is sent to queue $(k+1,j)$ at time $n$, i.e.,
\begin{align*}
    s^n_{k,j \rightarrow j} =
     \begin{cases}
     1 & \text{processed type-$k$ task in $(k,j)$ is sent to}\\ &\text{$(k+1,j)$ at time $n$}\\
     0 & \text{otherwise}.\\
     \end{cases}
\end{align*}
We define $\vec{u}^{\,n}$, $\vec{w}^{\,n}$ and $\vec{s}^{\,n}$ to be the decision vectors for $u^n_{m \rightarrow j}$, $w^n_{k,j \rightarrow l}$ and $s^n_{k,j \rightarrow j}$ respectively at time $n$ .

Now, we state the dynamics of the queueing network. If $k \in \mathcal{C}$, then
\begin{equation}
    Q^{n+1}_{(k,j)}=Q^n_{(k,j)} + a^n_{m(k)}u^n_{m(k) \rightarrow j} -d^n_{(k,j)};
\end{equation}
else,
\begin{align}
    Q^{n+1}_{(k,j)}=Q^n_{(k,j)} &+ \sum_{l \in [J] \backslash \{j\}}d^n_{(k-1,l),c}w^n_{k-1,l \rightarrow j} \nonumber \\ 
    & + d^n_{(k-1,j)}s^n_{k-1,j \rightarrow j} -d^n_{(k,j)}.
\end{align}
For $\forall \ k \in [K] \backslash \mathcal{H}$,
\begin{equation}
Q^{n+1}_{(k,j),c} = Q^n_{(k,j),c} + d^n_{(k,j)}(1-s^n_{k,j \rightarrow j}) - d^n_{(k,j),c}.
\end{equation}
\subsection{Queueing Network Planning Problem}
Before we introduce an optimization problem that characterizes the capacity region of the described  queueing network, we first define the following terms. 

Consider an arbitrary scheduling policy. Let $u_{m \rightarrow j}$ be the long-run fraction that root tasks of the type-$m$ job are scheduled on server $j$. We define $\vec{u}$ to be the root-node allocation vector. A root-node allocation vector $\vec{u}$ is feasible if    
\begin{equation} \label{eq:u}
    \sum^J_{j=1} u_{m \rightarrow j} =1, \quad \forall \ m \in [M].
\end{equation}
For the type-$k$ tasks served by server $j$, we denote $s_{k,j\rightarrow j}$ as the long-run fraction that their child tasks (type-$(k+1)$) are scheduled on server $j$. Then, we denote  $\vec{s}$ as an allocation vector. 

For the outputs of processed type-$k$ tasks in virtual queue $(k,j),c$, we denote $w_{k,j \rightarrow l}$ be the long-run fraction that they are sent to virtual queue $(k+1,l)$. An allocation vector $\vec{w}$ is feasible if
\begin{equation}\label{eq:w_lj}
    \sum_{l \in [J] \backslash \{j\}} w_{k,j \rightarrow l}= 1, \quad \forall \ j \in [J], \ \forall \ k \in [K] \backslash \mathcal{H}.
\end{equation}
For the type-$k$ tasks which are processed by server $j$, we define $f_{k,j \rightarrow l}$ as long-run fraction that their child tasks are scheduled on server $l$. Given allocation vectors $\vec{s}$ and $\vec{w}$, we can write $f_{k,j \rightarrow l}$ as follows:
\begin{align} \label{eq:f}
f_{k,j \rightarrow l} = 
    \begin{cases}
    s_{k,j \rightarrow j}, \ & \text{if} \ l = j\\
    (1-s_{k,j \rightarrow j})w_{k,j \rightarrow l}, \ & \text{otherwise}.
    \end{cases}
\end{align}
Clearly, we have 
\begin{align}
    \sum^J_{l=1} f_{k,j \rightarrow l}= 1, \quad \forall \  j \in [J], \ \forall \ k \in [K] \backslash \mathcal{H}.
\end{align}
Let $r_{(k,j)}$ denote the nominal rate to the virtual queue $(k,j)$ and $r_{(k,j),c}$ denote the nominal rate to the virtual queue $(k,j)_c$. If $k \in \mathcal{C}$, the nominal rate $r_{(k,j)}$ can be written as
\begin{equation} \label{eq:rkj}
    r_{(k,j)} = \lambda_{m(k)}u_{m(k) \rightarrow j}.
\end{equation}
If $k \notin \mathcal{C}$, the rate $r_{(k,j)}$ can be obtained by summing $r_{(k-1,l)}$ with $f_{k-1,l\rightarrow j}$ over all servers, i.e., \begin{equation} \label{eq:rkj_notroot}
r_{(k,j)} = \sum^J_{l=1} r_{(k-1,l)} f_{k-1,l \rightarrow j},
\end{equation}
because of flow conservation. Moreover, the rate $r_{(k,j),c}$ can be written as 
\begin{equation}\label{eq:rkjc}
r_{(k,j),c}  = r_{(k,j)}(1-s_{k,j \rightarrow j}).
\end{equation}

Now, we introduce an optimization problem called \textit{queueing network planning problem (QNPP)} that characterizes the capacity region of the virtual queueing network. Given the arrival rate vectors $\vec{\lambda}$, the queueing network planning problem ensures that the service rate allocated to each queue in the queueing network is at least as large as the nominal traffic rate to that queue. The problem is defined as follows:

\textbf{Queueing Network Planning Problem (QNPP):}
\begin{align}
& \text{Maximize} \quad  \gamma \label{eq:op}\\
& \text{subject to} \quad  r_{(k,j)} \leq \mu_{(k,j)}p_{(k,j)} - \gamma, \ \forall \ j, \ \forall \ k.\\
& r_{(k,j),c} \leq \frac{b_jq_{(k,j)}}{c_k} - \gamma, \ \forall \ j, \  \forall \ k\in [K] \backslash \mathcal{H}.
\end{align}
and subject to allocation vectors being feasible, where $r_{(k,j)}$ and $r_{(k,j),c}$ are nominal rate defined in (\ref{eq:rkj})-(\ref{eq:rkjc}). Note that the allocation vectors $\vec{p}$, $\vec{q}$, $\vec{u}$, $\vec{w}$ are feasible if (\ref{eq:p}), (\ref{eq:q}), (\ref{eq:u}) and (\ref{eq:w_lj}) are satisfied. Based on QNPP above, the capacity region of the virtual queueing network can be characterized by following proposition.
\begin{proposition}\label{optimal_QNPP}
The capacity region $\Lambda^{'}$ of the virtual queueing network characterizes the set of all rate vectors $\vec{\lambda} \in \mathbb{R}^M_{+}$ for which the corresponding optimal solution $\gamma^*$ to the queueing network planning problem (QNPP) satisfies $\gamma^* \geq 0$. In other words, \textit{capacity region} $\Lambda^{'}$ is characterized as follows 
\begin{align*}
   & \Lambda^{'} \triangleq \Bigg\{\vec{\lambda} \in \mathbb{R}^{M}_{+}: \exists \ \vec{p} \geq \vec{0}, \ \vec{q}\geq \vec{0} , \ \vec{u} \geq \vec{0},\ \vec{s} \geq \vec{0}, \ \vec{w} \geq \vec{0} \\ & \text{s.t.} \sum^{K}_{k=1}p_{(k,j)} \leq 1 \ \forall \ j, \ r_{(k,j),c} \leq \frac{b_jq_{(k,j)}}{c_k} \; \forall \ j, \ \forall \ k \in [K] \backslash \mathcal{H},\\
    & \sum_{k \in [K] \backslash \mathcal{H}}q_{(k,j)} \leq 1 \ \forall \ j, \  r_{(k,j)} \leq \mu_{(k,j)}p_{(k,j)} \ \forall \ k, \\
    & \sum^J_{j=1} u_{m \rightarrow j} =1 \ \forall \ m, \sum_{l \in [J]\backslash \{j\} } w_{k,j \rightarrow l}= 1 \ \forall \ j, \ \forall \ k \in [K] \backslash \mathcal{H}\Bigg\}.
  \end{align*}
 \end{proposition}
\begin{proof}
 We consider the virtual queueing network in the fluid limit. Define the amount of fluid in virtual queue corresponding to type-$k$ tasks processed by server $j$ as  $X_{(k,j)}(t)$. Similarly, define the amount of fluid in virtual queue corresponding to processed type-$k$ tasks sent by server $j$ as $X_{(k,j),c}(t)$. If $k \in \mathcal{C}$, the dynamics of the fluid are as follows
\begin{equation}
    X_{(k,j)}(t) = X_{(k,j)}(0) + \lambda_{m(k)}u_{m(k) \rightarrow j}t-\mu_{(k,j)}p_{(k,j)}t;
\end{equation}
else,
\begin{align}
X_{(k,j)}(t) & = X_{(k,j)}(0) + (\sum_{l \in [J] \backslash \{j\}}\frac{b_lq_{(k-1,l)}}{c_{k-1}}w_{k-1,l \rightarrow j} \nonumber \\
& +\mu_{(k-1,j)}p_{(k-1,j)}s_{k-1,j \rightarrow j}-\mu_{(k,j)}p_{(k,j)})t.
\end{align}
For $k \in [K] \backslash\mathcal{H}$, we have
\begin{equation}
X_{(k,j),c}(t) = X_{(k,j),c}(0) + (1-s_{k,j \rightarrow j})\mu_{(k,j)}p_{(k,j)}t - \frac{b_jq_{(k,j)}}{c_k}t.
\end{equation}
We suppose that $\gamma^{*}<0$. Let's show that the virtual queueing network is weakly unstable. In contrary, we suppose that there exists a scheduling policy such that under that policy for all $t \geq 0$, we have
\begin{align}
    & X_{(k,j)}(t)=0, \ \forall \ j, \ \forall \ k \\
    & X_{(k,j),c}(t)=0, \ \forall \ j, \ \forall \ k \in [K]\backslash \mathcal{H}.
\end{align}
Now, we pick a regular point $t_1$. Then, for all $k$ and all $j$, $\dot{X}_{(k,j)}(t_1) = 0$ which implies that 
there exit allocation vectors $\vec{p}$, $\vec{q}$, $\vec{u}$, $\vec{s}$, $\vec{w}$ such that 
\begin{align}
 &   \lambda_{m(k)}u_{m(k) \rightarrow j} - \mu_{(k,j)}p_{(k,j)} = 0, \ \forall \ j,\ \forall \ k \in \mathcal{C}\\
&\sum_{l \in [J] \backslash \{j\}}\frac{b_lq_{(k-1,l)}}{c_{k-1}}w_{k-1,l \rightarrow j } \nonumber +\mu_{(k-1,j)}p_{(k-1,j)}s_{k-1,j \rightarrow j} \\ &-\mu_{(k,j)}p_{(k,j)} = 0, \ \forall \  j, \ \forall \ k\in [K] \backslash \mathcal{C}.
\end{align}
Similarly, we have 
\begin{align}
    (1-s_{k,j \rightarrow j})\mu_{(k,j)}p_{(k,j)} - \frac{b_jq_{(k,j)}}{c_k} = 0 , \ \forall \  j, \ \forall \ k\in [K] \backslash \mathcal{H}. \nonumber
\end{align}
Now, we show that $r_{(k,j)} = \mu_{(k,j)}p_{(k,j)}$ for all $j$ and all $k$ by induction. First, consider an arbitrary $k \in \mathcal{C}$ which is the root node of job $m(k)$, then
\begin{align}
r_{(k,j)} = \lambda_{m(k)}u_{m(k) \rightarrow j} = \mu_{(k,j)}p_{(k,j)}, \ \forall \ j.
\end{align}
Then, we have 
\begin{align}
& r_{(k+1,j)} = \sum^J_{l=1} r_{(k,l)} f_{k,l \rightarrow j} = \sum^J_{l=1}\mu_{(k,l)}p_{(k,l)}f_{k,l \rightarrow j}\\
= &  \sum_{l \in [J] \backslash \{j\}}\mu_{(k,l)}p_{(k,l)}(1-s_{k,l \rightarrow l})w_{k,l \rightarrow j} + \mu_{(k,j)}p_{(k,j)}s_{k,j \rightarrow j}\\
= &  \sum_{l \in [J] \backslash \{j\}}\frac{b_lq_{(k,l)}}{c_{k}}w_{k,l \rightarrow j} + \mu_{(k,j)}p_{(k,j)}s_{k,j \rightarrow j}\\ =& \mu_{(k+1,j)}p_{(k+1,j)}.
\end{align}
By induction, we have $r_{(k,j)} = \mu_{(k,j)}p_{(k,j)}$ for all $j$ and all $k \in \mathcal{I}_m$. Thus, we aslo have 
\begin{align}
     r_{(k,j),c} & =   r_{(k,j)}(1-s_{k,j \rightarrow j}) = \mu_{(k,j)}p_{(k,j)}(1-s_{k,j \rightarrow j})\\& = \frac{b_jq_{(k,j)}}{c_k}, \ \forall \ j, \ \forall \ k \in [K]\backslash \mathcal{H}.
\end{align}
which contradicts $\gamma^{*}<0$.

Now, we suppose that $\gamma^{*} \geq 0$. It follows that there exist vectors $\vec{p}^{\,*}$, $\vec{q}^{\,*}$, $\vec{u}^{\,*}$, $\vec{s}^{\,*}$, $\vec{w}^{\,*}$, $\epsilon_{(k,j)}$ and $\epsilon_{(k,j),c}$ such that   
\begin{align}
& \mu_{(k,j)}p^*_{(k,j)} = r^*_{(k,j)}+\epsilon_{(k,j)}, \ \forall \ j,\ \forall \ k; \label{eq:opt_rkj}\\
& \frac{b_jq^*_{(k,j)}}{c_k} =r^*_{(k,j),c}+\epsilon_{(k,j),c}, \ \forall \  j, \ \forall \ k\in [K] \backslash \mathcal{H}; \label{eq:opt_rkjc}\\
& \epsilon_{(k,j)} = 
\frac{y[k-\sum^{m(k)-1}_{m^{'}=1}K_{m^{'}}-1]}{y[K_{m(k)}]}\gamma^{*}, \ \forall \ j, \ \forall \  k;\\
& \epsilon_{(k,j),c} = \epsilon_{(k,j)}+\frac{1}{y[K_{m(k)}]}\gamma^{*}, \ \forall \ j , \ \forall \ k\in [K] \backslash \mathcal{H}.
\end{align}
where the sequence $y[n]$ satisfies the recurrence relationship: \begin{align}
    y[1] = 1; \ y[n] = J(y[n-1] + 1), \forall \ n > 1.
\end{align}
Therefore, if $k \in \mathcal{C}$, then
\begin{align}
& \dot{X}^{*}_{(k,j)}(t) = \lambda_{m(k)}u^*_{m(k) \rightarrow j} - \mu_{(k,j)}p^*_{(k,j)} \label{eq:root_X}\\
 = & r^*_{(k,j)} - \mu_{(k,j)}p^*_{(k,j)} = - \epsilon_{(k,j)} \leq 0, \ \forall \ j, \ \forall \ t >0 \label{eq:derivative_Xkj_root};
\end{align}
else,
\begin{align}
& \dot{X}^{*}_{(k,j)}(t)  =  \sum_{l \in [J] \backslash \{j\}}\frac{b_lq^*_{(k-1,l)}}{c_{k-1}}w^*_{k-1,l \rightarrow j } \nonumber \\& +\mu_{(k-1,j)}p^*_{(k-1,j)}s^*_{k-1,j \rightarrow j}-\mu_{(k,j)}p^*_{(k,j)} \label{eq:notroot_X}\\
& = \sum_{l \in [J] \backslash \{j\}}r^{*}_{(k-1,l),c}w^*_{k-1,l \rightarrow j }+r^*_{(k-1,j)}s^*_{k-1,j \rightarrow j}-r^*_{(k,j)} \nonumber \\
& + \sum_{l \in [J] \backslash \{j\}}\epsilon_{(k-1,l),c}w^{*}_{k-1,l \rightarrow j}+\epsilon_{(k-1,j)}s^{*}_{k-1,j \rightarrow j} - \epsilon_{(k,j)}\label{eq:notroot_X2}.
\end{align}
Note that (\ref{eq:notroot_X}) to (\ref{eq:notroot_X2}) follows from (\ref{eq:opt_rkj}) and (\ref{eq:opt_rkjc}). In (\ref{eq:notroot_X2}), we have
\begin{align}
 & \sum_{l \in [J] \backslash \{j\}}  r^*_{(k-1,l)}(1-s^{*}_{k-1,l \rightarrow l})w^*_{k-1,l \rightarrow j} \nonumber \\
& +r^*_{(k-1,j)}f^*_{k-1,j \rightarrow j}-r^*_{(k,j)} \label{eq:notroot_X3}\\
& = \sum_{l \in [J] \backslash \{j\}} r^*_{(k-1,l)}f^*_{k-1,l \rightarrow j}+r^*_{(k-1,j)}f^*_{k-1,j \rightarrow j}-r^*_{(k,j)} \label{eq:notroot_X4} \\
& = \sum_{1 \leq l \leq J}r^*_{(k-1,l)}f^*_{k-1,l \rightarrow j}-r^*_{(k,j)} = 0,  \label{eq:notroot_X5}
\end{align}
and
\begin{align}
 &\sum_{l \in [J] \backslash \{j\}}\epsilon_{(k-1,l),c}w^{*}_{k-1,l \rightarrow j}+\epsilon_{(k-1,j)}s^{*}_{k-1,j \rightarrow j} - \epsilon_{(k,j)} \label{eq:notroot_X_epsilon1}\\
& \leq \sum_{l \in [J] \backslash \{j\}}\epsilon_{(k-1,l),c}+\epsilon_{(k-1,j)} - \epsilon_{(k,j)} \label{eq:notroot_X_epsilon2} = - \frac{\gamma^{*}}{y[K_{m(k)}]}.
\end{align}
Note that (\ref{eq:notroot_X3}) to (\ref{eq:notroot_X4}) follows from (\ref{eq:f}); and (\ref{eq:notroot_X4}) to (\ref{eq:notroot_X5}) follows from (\ref{eq:rkj_notroot}); (\ref{eq:notroot_X_epsilon1}) to (\ref{eq:notroot_X_epsilon2}) is because $w^*_{k-1,l \rightarrow j}\leq 1$ and $s^*_{k-1,j \rightarrow j} \leq 1$. Thus, (\ref{eq:notroot_X}) can be written as 
\begin{align}
    \dot{X}^{*}_{(k,j)}(t) \leq - \frac{\gamma^{*}}{y[K_{m(k)}]} \leq 0, \ \forall \ j, \ \forall \ t >0 \label{eq:derivative_X_kj}.
\end{align}
For $k \in [K] \backslash \mathcal{H}$, we have
\begin{align}
&\dot{X}^{*}_{(k,j),c}(t) = (1-s^*_{k,j \rightarrow j})\mu_{kj}p^*_{(k,j)} - \frac{b_jq^*_{(k,j)}}{c_k} \label{eq:com_X}\\
                  & = (1-s^*_{k,j \rightarrow j})(r^*_{(k,j)} +\epsilon_{(k,j)})- r^*_{(k,j),c}-\epsilon_{(k,j),c}\\
                  & = (1-s^*_{k,j \rightarrow j})\epsilon_{(k,j)}-\epsilon_{(k,j),c}\\
                  & \leq \epsilon_{(k,j)}-\epsilon_{(k,j),c}
                  = - \frac{\gamma^{*}}{y[K_{m(k)}]} \leq 0 \label{eq:derivative_X_kjc}.
\end{align}
If $X_{(k,j)}(0)=0$ for all $k$ and all $j$, then $X_{(k,j)}(t)=0$ for all $t\geq 0$, all $j$ and all $k$. Also, if $X_{(k,j),c}(0)=0$ for all $k \in [K] \backslash \mathcal{H}$ and all $j$, then $X_{(k,j),c}(t)=0$ for all $t\geq 0$, all $j$ and all $k$. Thus, the virtual queueing network is weakly stable, i.e., the queueing network process is rate stable.
 \end{proof}
\section{Throughput-Optimal Policy}\label{sec:optimal}
In this section, we propose Max-Weight scheduling policy for the network of virtual queues in Section \ref{sec:queue} and show that it is throughput-optimal for the network of the original scheduling problem. 

The proposed virtual queueing network is quite different from traditional queueing networks since the proposed network captures the communication procedures (routing of tasks determined by scheduling policies) in the network. Therefore, it is not clear that the capacity region characterized by QNPP is equivalent to the capacity region characterized by SPP. To prove the throughput-optimality of Max-Weight policy for the original scheduling problem, we first need to show the equivalence of capacity regions characterized by SPP and QNPP. Then, under the Max-Weight policy, we consider the queueing network in the fluid limit, and using a Lyapunov argument, we show that the fluid model of the virtual queueing network is weakly stable for all arrival vectors in the capacity region, and stable for all arrival vectors in the interior of the capacity region. 

Now, we give a description of the Max-Weight policy for the proposed virtual queueing network. Given virtual queue-lengths $Q^n_{(k,j)}$, $Q^n_{(k,j),c}$ and history $\mathcal{F}^n$ at time $n$, Max-Weight policy allocates the vectors $\vec{p}$, $\vec{q}$, $\vec{u}$, $\vec{s}$ and $\vec{w}$ that are\footnote{We define $\mathcal{F}^n$ to be the $\sigma$-algebra generated by all the random variables in the system up to time $n$.}
\begin{align*}
&\arg \max_{\vec{p}, \vec{q}, \vec{u}, \vec{s}, \vec{w}} -(\vec{Q}^{\,n})^{T} E[\Delta \vec{Q}^{\,n}|\mathcal{F}^n]-(\vec{Q}^{\,n}_c)^{T} E[\Delta \vec{Q}^{\,n}_c|\mathcal{F}^n]\\
=& \arg \min_{\vec{p}, \vec{q}, \vec{u}, \vec{s}, \vec{w}} (\vec{Q}^{\,n})^{T} E[\Delta \vec{Q}^{\,n}|\mathcal{F}^n]+(\vec{Q}^{\,n}_c)^{T} E[\Delta \vec{Q}^{\,n}_c|\mathcal{F}^n],
\end{align*}
where $\vec{Q}^{\,n}$ and $\vec{Q}^{\,n}_c$ are the vectors of queue-lengths $Q^n_{(k,j)}$ and $Q^n_{(k,j),c}$ at time $n$. The Max-Weight policy is the choice of $\vec{p}$, $\vec{q}$, $\vec{u}$, $\vec{s}$ and $\vec{w}$ that minimizes the drift of a Lyapunov function $V^n=\sum_{k,j}(Q^n_{(k,j)})^2+\sum_{k,j}(Q^n_{(k,j),c})^2$.

The following theorem shows the throughput-optimality of Max-Weight policy. 
\begin{theorem}
\label{optimal}
Max-Weight policy is throughput-optimal for the network, i.e. Max-Weight policy is rate stable for all the arrival vectors in the capacity region $\Lambda$ defined in Proposition \ref{optimal_LP}, and it makes the underlying Markov process positive recurrent for all the arrival rate vectors in the interior of $\Lambda$. 
\end{theorem}
\begin{proof}In order to prove Theorem 1, we first state Lemma \ref{equivalence} which is proved in Appendix \ref{appendix:lemma1}.
\begin{lemma}
\label{equivalence}
The capacity region characterized by static planning problem (SPP) is equivalent to the capacity region characterized by queueing network planning problem (QNPP), i.e. $\Lambda=\Lambda^{'}$.
\end{lemma}
Having Lemma \ref{equivalence}, we now show that the queueing network is rate stable for all arrival vectors in $\Lambda^{'}$, and strongly stable for all arrival vectors in the interior of $\Lambda^{'}$ under Max-Weight policy. 

We consider the problem in the fluid limit. Define the amount of fluid in virtual queue corresponding to type-$k$ tasks  by server $j$ as  $X_{(k,j)}(t)$. Similarly, define the amount of fluid in virtual queue corresponding to processed type-$k$ tasks sent by server $j$ as $X_{(k,j),c}(t)$.

Now, we define $\gamma^{*}$ as the optimal value of QNPP. If we consider a rate vector $\vec{\lambda}$ in the interior of $\Lambda^{'}$, then $\gamma^{*}>0$. Directly from (\ref{eq:derivative_Xkj_root}), (\ref{eq:derivative_X_kj}) and (\ref{eq:derivative_X_kjc}), for $t>0$, we have
\begin{align}
    & \dot{X}^{*}_{(k,j)}(t)  < 0, \ \forall \ j, \forall \ k; \label{eq:X_kj_negative}\\
    & \dot{X}^{*}_{(k,j),c}(t)  < 0, \ \forall \ j, \ \forall \ k \in [K]\backslash \mathcal{H}. \label{eq:X_kjc_negative}
\end{align}
Now, we take $V(t)=\frac{1}{2}\vec{X}^T(t)\vec{X}(t)+\frac{1}{2}\vec{X}^T_c(t)\vec{X}_c(t)$ as the Lyapunov function where $X(t)$ and $X_c(t)$ are vectors for $X_{(k,j)}(t)$ and $X_{(k,j),c}(t)$ respectively. The drift of $V$ by using Max-Weight policy is
\begin{align}
  & \dot{V}_{\text{Max-Weight}}(t) =  \min_{\vec{p}, \vec{q}, \vec{u}, \vec{s}, \vec{w}} \vec{X}^T(t)\vec{\dot{X}}(t) + \vec{X}_c^T(t)\vec{\dot{X}}_c(t) \label{eq:maxweight}\\
  & \leq (\vec{X}^{*})^T(t)(\vec{\dot{X}}^*)(t)+(\vec{X}^*_c)^T(t)\vec{\dot{X}}^*_c(t) < 0, \ \forall \ t >0, 
\end{align}
using (\ref{eq:X_kj_negative}) and (\ref{eq:X_kjc_negative}). Thus, for $t >0$, we show that $\dot{V}_{\text{Max-Weight}}(t) < 0$ if $\vec{\lambda}$ in the interior of $\Lambda^{'}$. This proves that the fluid model is stable under Max-Weight policy, which implies the positive recurrence of the underlying Markov chain \cite{dai1995positive}.

Consider a vector $\vec{\lambda} \in \Lambda^{'}$, there exist allocation vectors $\vec{p}^{\,*}$, $\vec{q}^{\,*}$, $\vec{u}^{\,*}$, $\vec{s}^{\,*}$ and $\vec{w}^{\,*}$ such that
\begin{align}
& \mu_{(k,j)}p^*_{(k,j)} = r^*_{(k,j)}, \ \forall \ j,\ \forall \ k.\\
& \frac{b_jq^*_{(k,j)}}{c_k} =r^*_{(k,j),c}, \ \forall \  j, \ \forall \ k\in [K] \backslash \mathcal{H}
\end{align}
From (\ref{eq:root_X}) and (\ref{eq:notroot_X}), we have 
\begin{align}
    \dot{X}^{*}_{(k,j)}(t) = 0, \ \forall \ j, \ \forall \ k, \ \forall \ t >0.
\end{align}
From (\ref{eq:com_X}), we have 
\begin{align}
    \dot{X}^{*}_{(k,j),c}(t) = 0, \ \forall \ j, \ \forall \ k\in [K] \backslash \mathcal{H}, \ \forall \ t >0. 
\end{align}
From (\ref{eq:maxweight}), for $t >0$, we have $\dot{V}_{\text{Max-Weight}}(t) \leq 0$ if $\vec{\lambda} \in \Lambda^{'}$. If $\vec{X}(0)=\vec{0}$ and $\vec{X}_c(0)=\vec{0}$, we have $\vec{X}(t)=\vec{0}$ and $\vec{X}_c(t)=\vec{0}$ for all $t \geq 0$, which shows that the fluid model is weakly stable under Max-Weight Policy, i.e., the queueing network process is rate stable. Therefore, Max-Weight policy is throughput-optimal for the queueing network which completes the proof.
 \end{proof}
\begin{remark}
In the proof of Theorem \ref{optimal}, the most significant part is to prove Lemma \ref{equivalence} which shows that capacity region of the original scheduling problem (characterized by a LP) is equivalent to the capacity region of the proposed virtual queueing network (characterized by a complicated mathematical optimization problem). In \cite{pedarsani2014scheduling,pedarsani2017robust}, without communication constraints, the capacity regions of original problem (denoted by $\tilde{\Lambda}$) and the virtual queueing network (denoted by $\tilde{\Lambda}^{'}$) are characterized by LPs. Given a $\vec{\lambda} \in \tilde{\Lambda}$ with corresponding allocation vector, one can construct a feasible allocation vector for virtual queueing network supporting $\vec{\lambda}$ by splitting the traffic equally from a queue into the following branching queues. In Lemma \ref{equivalence}, to prove $\Lambda \subseteq \Lambda^{'}$, we construct feasible vectors for virtual queueing network supporting $\vec{\lambda}$ by splitting traffic differently from a queue into following branching queues through a careful and clever design, which is fundamentally different from \cite{pedarsani2014scheduling,pedarsani2017robust}.
\end{remark}
\section{Complexity of throughput-optimal policy}\label{sec:complexity}
In the previous section, we showed the throughput-optimality of the scheduling policy, but it is not clear how complex it is to implement the policy. In this section, we describe how the Max-Weight policy is implemented and show that the Max-Weight policy has complexity which is almost linear in the number of virtual queues.  

First, we denote $p^n_{(k,j)} \in \{0,1\}$ as the decision variable that server $j$ processes task $k$ in $(k,j)$ at time $n$, i.e.,
\begin{align*}
    p^n_{(k,j)} =
     \begin{cases}
     1 & \text{server $j$ processes task $k$ in $(k,j)$ at time $n$}\\
     0 & \text{otherwise}.\\
     \end{cases}
\end{align*}
We denote $q^n_{(k,j)} \in \{0,1\}$ as the decision variable that server $j$ sends the data of processed task $k$ in $(k,j),c$ at time $n$, i.e.,
\begin{align*}
    q^n_{(k,j)} =
     \begin{cases}
     1 & \text{server $j$ sends the data of processed task $k$}\\
     & \text{in $(k,j),c$ at time $n$}\\
     0 & \text{otherwise}.\\
     \end{cases}
\end{align*}
Then, we define $\vec{p}^{\,n}$ and $\vec{q}^{\,n}$ to be decision vectors for $p^n_{(k,j)}$ and $q^n_{(k,j)}$ respectively at time $n$.

Given virtual queue-lengths $\vec{Q}^{\,n}$, $\vec{Q}^{\,n}_c$ and history $\mathcal{F}^n$ at time $n$, Max-Weight policy minimizes 
\begin{equation}
    (\vec{Q}^{\,n})^{T} E[\Delta \vec{Q}^{\,n}|\mathcal{F}^n]+(\vec{Q}^{\,n}_c)^{T} E[\Delta \vec{Q}^{\,n}_c|\mathcal{F}^n]
\end{equation}
over the vectors $\vec{p}^{\,n}$, $\vec{q}^{\,n}$, $\vec{u}^{\,n}$, $\vec{s}^{\,n}$ and $\vec{w}^{\,n}$. That is, Max-Weight policy minimizes 
\begin{align}
     & \sum^J_{j=1} \sum_{k \in \mathcal{C}}Q^n_{(k,j)}(\lambda_{m(k)}u^n_{m(k) \rightarrow j}-\mu_{(k,j)}p^n_{(k,j)}) \nonumber \\
 & +  \sum^J_{j=1} \sum_{k \notin \mathcal{C}} Q^n_{(k,j)}(\sum_{l \in [J] \backslash \{j\}}\frac{b_lq^n_{(k-1,l)}}{c_{k-1}}w^n_{k-1,l \rightarrow j} \nonumber \\
 & +\mu_{(k-1,j)}p^n_{(k-1,j)}s^n_{k-1,j \rightarrow j}-\mu_{(k,j)}p^n_{(k,j)}) \nonumber\\
  & +  \sum^J_{j=1} \sum_{k \in [K]\backslash \mathcal{H}}Q^n_{(k,j),c} \{(1-s^n_{k,j \rightarrow j})\mu_{(k,j)}p^n_{(k,j)} -\frac{b_jq^n_{(k,j)}}{c_k}\} \nonumber
  \end{align}
which can be rearranged to 
 \begin{align}
& \sum^J_{j=1}\sum_{k \in [K]\backslash \mathcal{H}}q^n_{(k,j)}\frac{b_j}{c_k}(\sum_{l \in [J] \backslash \{j\}}w^n_{k,j \rightarrow l}Q^n_{(k+1,l)} -Q^n_{(k,j),c}) \nonumber\\ & + \sum^J_{j=1} \sum_{k \in [K]\backslash \mathcal{H}} p^n_{(k,j)} \mu_{(k,j)} \{s^n_{k,j \rightarrow j}(Q^n_{(k+1,j)} -Q^n_{(k,j),c}) \nonumber \\ &+Q^n_{(k,j),c}-Q^n_{(k,j)}\} - \sum^J_{j=1}\sum_{k \in \mathcal{H}} p^n_{(k,j)} \mu_{(k,j)} Q^n_{(k,j)} \nonumber \\& + \sum_{k\in \mathcal{C}} \lambda_{m(k)} \sum^J_{j=1} Q^n_{(k,j)}u^n_{m(k) \rightarrow j} \nonumber\\
= &\sum^J_{j=1}\sum^K_{k=1} p^n_{(k,j)}F^n_{(k,j)} + \sum^J_{j=1}\sum_{k \in [K]\backslash \mathcal{H}} q^n_{(k,j)}G^n_{(k,j)} \nonumber \\ & + \sum_{k\in \mathcal{C}} \lambda_{m(k)} \sum^J_{j=1} Q^n_{(k,j)}u^n_{m(k) \rightarrow j}. \label{eq:objective_2}
\end{align}
Note that the function $F^n_{(k,j)}$ is defined as follows: 
\begin{align}
F^n_{(k,j)} & = \mu_{(k,j)} \{s^n_{k,j \rightarrow j}(Q^n_{(k+1,j)}-Q^n_{(k,j),c}) \nonumber \\&+Q^n_{(k,j),c}-Q^n_{(k,j)}\}, \ \forall \ j, \ \forall \ k \in [K] \backslash  \mathcal{H}; 
\end{align}
else
\begin{align}
    F^n_{(k,j)} = - \mu_{(k,j)} Q^n_{(k,j)}, \ \forall \ j, \ \forall \ k \in \mathcal{H}.
\end{align}
The function $G^n_{(k,j)}$ is defined as follows:
\begin{align}
G^n_{(k,j)}= & \frac{b_j}{c_k}(\sum_{l \in [J] \backslash \{j\}}w^n_{k,j \rightarrow l}Q^n_{(k+1,l)}-Q^n_{(k,j),c}), \nonumber \\ 
& \forall \ k \in [K]\backslash \mathcal{H}.
\end{align}
For (\ref{eq:objective_2}), we first minimize 
\begin{align}
    \sum_{k\in \mathcal{C}} \lambda_{m(k)} \sum^J_{j=1} Q^n_{(k,j)}u^n_{m(k) \rightarrow j} \label{eq:U}
\end{align}
 over $\vec{u}^{\,n}$. We denote $j^n_u(k) = \arg \min_{j \in [J]}Q^n_{(k,j)}$, $\forall \ k \in \mathcal{C}$. It is clear that the minimizer $\vec{u}^{\,n*}$ is  
\begin{align}
    u^{n*}_{m(k) \rightarrow j}=
    \begin{cases}
    1 \ \text{if} \ j = j^n_u(k) \\
    0  \ \text{otherwise}.
    \end{cases}
\end{align}
When more than one queue-lengths are minima, we use a random tie-breaking rule.
Secondly, we minimize
\begin{align}
    \sum^J_{j=1}\sum^K_{k=1} p^n_{(k,j)}F^n_{(k,j)} \label{eq:F}
\end{align}
over $\vec{p}^{\,n}$ and $\vec{s}^{\,n}$. For each $j \in [J]$ and $k \in [K]\backslash \mathcal{H}$, the minimizer $s^{*}_{k,j \rightarrow j}$ of the function $F^n_{(k,j)}$ is 
\begin{align}
s^{n*}_{k,j \rightarrow j}=
\begin{cases}
1 \ \text{if} \ Q^n_{(k+1,j)}-Q^n_{(k,j),c} \leq 0\\
0 \ \text{if} \ Q^n_{(k+1,j)}-Q^n_{(k,j),c} >0.
\end{cases}
\end{align}
We denote $F^{n*}_{(k,j)}$ as minima of $F^n_{(k,j)}$ for each $j \in [J]$ and $k \in [K]$. Then, we define $k^n_F(j)$ as follows:
\begin{align}
    k^n_F(j) = & \arg\min_{k \in [K]} F^{n*}_{(k,j)}, \ \forall \ j \in [J]. \label{kf}
\end{align}
To minimize (\ref{eq:F}), we have 
\begin{align}
p^{n*}_{(k,j)}=
\begin{cases}
1 \ \text{if} \ k = k^n_F(j)\\
0 \ \text{otherwise}
\end{cases}
\end{align}
for each $j \in [J]$.
Lastly, we minimize 
\begin{align}
    \sum^J_{j=1}\sum_{k \in [K]\backslash \mathcal{H}}q^n_{(k,j)}G^n_{(k,j)} \label{eq:G}
\end{align}
over $\vec{q}^{\,n}$ and $\vec{w}^{\,n}$. We denote $j^n_w(k) = \arg \min_{l \in [J] \backslash \{j\}}Q^n_{(k+1,l)}$, $\forall \ k \in [K] \backslash \mathcal{H}$. For each $j \in [J]$ and $k \in [K]\backslash \mathcal{H}$, the minimizer of $G^n_{(k,j)}$ is 
\begin{align}
w^{n*}_{k,j \rightarrow l}=
\begin{cases}
1 \ \text{if} \ l = j^n_w(k)\\
0 \ \text{otherwise}.
\end{cases}
\end{align}
We denote $G^{n*}_{(k,j)}$ as minima of $G^n_{(k,j)}$ for each $j \in [J]$ and $k \in [K]\backslash \mathcal{H}$. Then, we define $k^n_G(j)$ as follows:
\begin{align}
    k^n_G(j) = & \arg\min_{k \in [K] \backslash \mathcal{H}} G^{n*}_{(k,j)}, \ \forall \ j \in [J]. \label{kg}
\end{align}
To minimize (\ref{eq:G}), we have
\begin{align}
q^{n*}_{(k,j)}=
\begin{cases}
1 \ \text{if} \ k = k^n_G(j)\\
0 \ \text{otherwise}
\end{cases}
\end{align}
for each $j \in [J]$.
Based on the optimization above, we describe how the Max-Weight policy is implemented. We consider the virtual queueing network at each time $n$. The procedures of minimizing (\ref{eq:U}) show that when a new job comes to the network, the root task $k$ is sent to virtual queue $(k,j)$ if queue length of $(k,j)$ is the shortest. The procedures of minimizing (\ref{eq:F}) show that server $j$ processes task $k$ in virtual queue $(k,j)$ if $p^{n*}_{(k,j)}=1$; and then the output of processed task $k$ is sent to virtual queue $(k+1,j)$ if $s^{n*}_{k,j \rightarrow j}=1$ or sent to virtual queue $(k,j),c$ if $s^{n*}_{k,j \rightarrow j}=0$. The procedures of minimizing (\ref{eq:G}) shows that server $j$ sends output of processed task $k$ in virtual queue $(k,j),c$ to virtual queue $(k+1,j^n_w(k))$ iff $q^{n*}_{(k,j)}=1$ and $w^{n*}_{k,j \rightarrow j^n_w(k)}=1$.
 By the analysis above, we know that the complexity of the Max-Weight policy is dominated by the procedure of sorting some linear combinations of queue lengths. As we know, the complexity of sorting algorithm is $N \log{N}$. To minimize (\ref{eq:objective_2}), there are $M$ procedures of sorting $J$ values for $Q^n_{(k,j)}$ if $k \in \mathcal{C}$, $K-M$ procedures of sorting $J-1$ values for $Q^n_{(k+1,l)}$, $J$ procedures of sorting $K-M$ values for $G^{n*}_{(k,j)}$, and $J$ procedures of sorting $K$ values for $F^{n*}_{(k,j)}$. Thus, the complexity of the Max-Weight policy is bounded above by $2KJ \log{K}+ KJ \log{J}$ which is almost linear in the number of virtual queues. Note that the number of the virtual queues in the network proposed in Section \ref{sec:queue} is $(2K-M)J$.
\section{Towards more general Computing Model}\label{sec:dag}
 \begin{figure}[t]
    \centering
    \includegraphics[width=\columnwidth]{DAG_and_network.pdf}
    \caption{Overview of DAG scheduling for dispersed computing.}\label{fig:DAG}
\end{figure}
In this section, we extend our framework to a more general computing model, where jobs are modeled as \textit{directed acyclic graphs (DAG)}, which capture more complicated dependencies among tasks. As shown in Fig. \ref{fig:DAG}, nodes of the DAG represent tasks of the job and edges represent the logic dependencies between tasks. Different from the model of chains, tasks in the model of DAGs might have more than one parent node, e.g., task $4$ in Fig. \ref{fig:DAG} has two parent nodes, task $2$ and task $3$.

One major challenge in the communication-aware DAG scheduling is that the data of processed parents tasks have to be sent to the same server for processing the child task. This \textit{logic dependency} difficulty for communications doesn't appear in the model of chains for computations because tasks of a chain can be processed one by one without the need for merging the processed tasks. Due to logic dependency difficulty incurred in the model of DAGs, designing a virtual queueing network which encodes the state of the network is more difficult.

Motivated by broadcast channel model in many areas (e.g. wireless network, message passing interface (MPI) and shared bus in computer networks), we simplify the network to a \textit{broadcast network}\footnote{Note that the communication-aware DAG scheduling in a general network without broadcast constraints would become intractable since all the tasks have to be tracked in the network for the purpose of merging, which can highly increase the complexity of scheduling policies.} in which the servers always broadcast the output of processed tasks to other servers. Inspired by \cite{pedarsani2014scheduling}, we propose a novel virtual queueing network to solve the logic dependency difficulty for communications, i.e., guarantee that the results of preceding tasks will be sent to the same server. Lastly, we propose the Max-Weight policy and show that it is throughput-optimal for the network.  
\subsection{DAG Computing Model}
For the DAG scheduling problem, we define the following terms. As shown in Fig. \ref{fig:DAG}, each job is modeled as a DAG. Let $(\mathcal{V}_m,\mathcal{E}_m,\{c_k\}_{k \in \mathcal{V}_m})$ be the DAG corresponding to the job of type $m$, $m \in [M]$, where $\mathcal{V}_m$ denotes the set of nodes of type-$m$ jobs, $\mathcal{E}$ represents the set of edges of the DAG, and $c_k$ denotes the data size (bits) of output type-$k$ task. Similar to the case of jobs modeled as chains, let the number of tasks of a type-$m$ job be $K_m$, i.e. $|\mathcal{V}_m| = K_m$, and the total number of task types in the network be $K$, and we index the task types in the network by $k$, $k \in [K]$, starting from job type $1$ to $M$. We call task $k'$ a parent of task $k$ if they belong to the same DAG, and $(k', k) \in \mathcal{E}_m$. Let $\mathcal{P}_k$ denote the set of parents of $k$. In order to process $k$, the processing of all the parents of $k$, $k' \in \mathcal{P}_k$, should be completed, and the results of the processed tasks should be all available in one server. We call task $k'$ a descendant of task $k$ if they belong to the same DAG,
and there is a directed path from $k$ to $k'$ in that DAG. 

In the rest of this section, we consider the network of dispersed computing platform to be the broadcast network where each server in the network always broadcasts the result of a processed task to other servers after that task is processed.

\subsection{Queueing Network Model for DAG Scheduling}\label{subsec:DAG_queue}
In this subsection, we propose a virtual queueing network that guarantees the output of processed tasks to be sent to the same server. Consider $M$ DAGs, $(\mathcal{V}_m,\mathcal{E}_m,\{c_k\}_{k \in \mathcal{V}_m})$, $m \in [M]$. We construct $M$ parallel networks of virtual queues by forming two kinds of virtual queues as follows:
\begin{enumerate}
   \item \textbf{Processing Queue}: 
      For each non-empty subset $\mathcal{S}_m$ of $\mathcal{V}_m$, $\mathcal{S}_m$ is a \emph{stage} of job $m$ if and only if for all $k \in \mathcal{S}_m$, all the descendants of $k$ are also in $\mathcal{S}_m$. For each stage of job $m$, we maintain one virtual queue. Also, a task $k$ in a stage is \textit{processable} if there are no parents of task $k$ in that stage.
    \item \textbf{Communication Queue}: For each server $j$ and each processable task $k$ in stage $\mathcal{S}_m$, we maintain one virtual queue which is indexed by $\mathcal{S}_m$ and $(k,j)$. 
\end{enumerate}
\begin{figure}[t]
    \centering
    \includegraphics[width=\columnwidth]{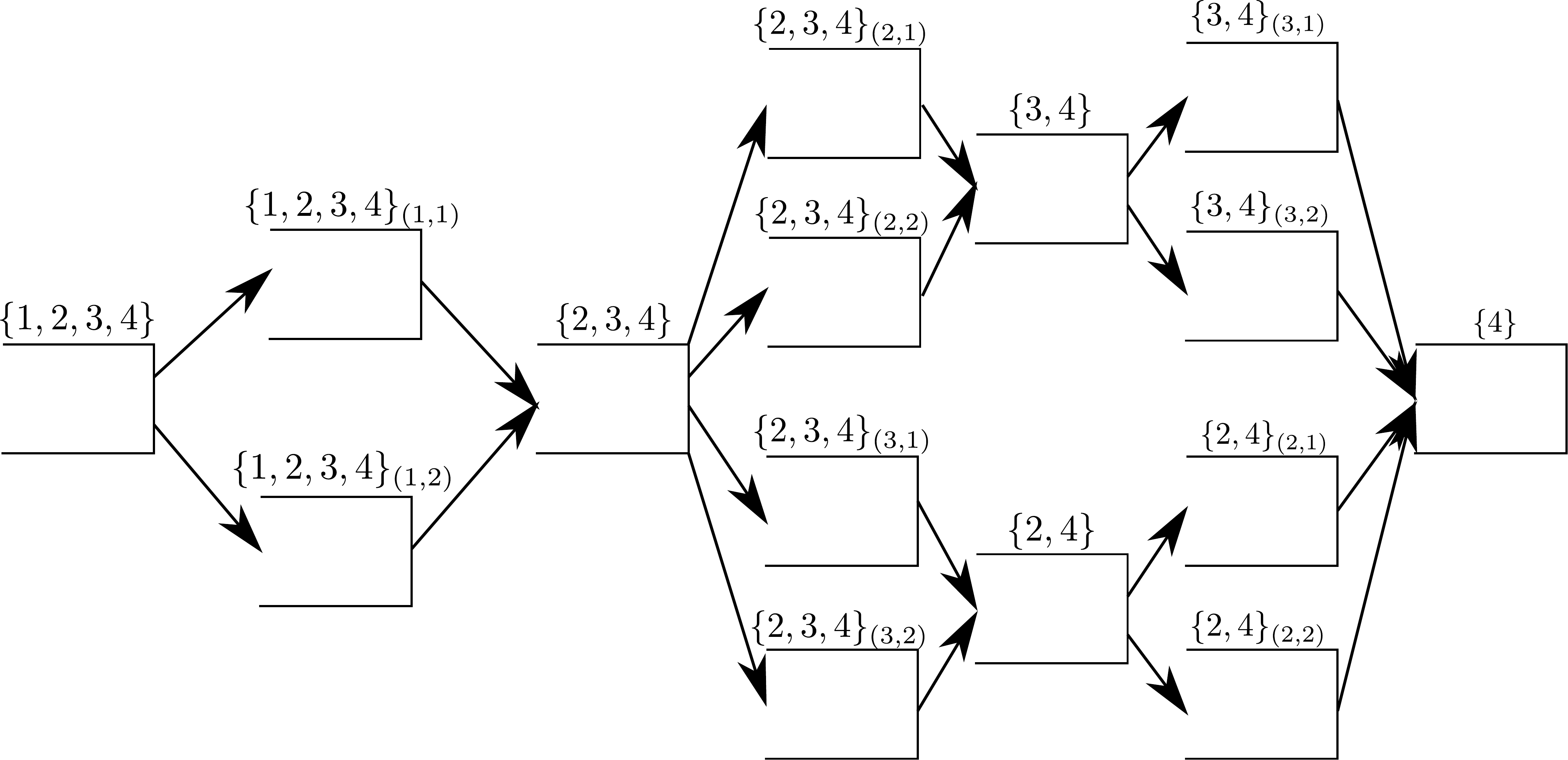}
    \caption{Queueing Network for the simple DAG in Fig. \ref{fig:DAG}.}\label{fig:DAG_queue}
\end{figure}
\begin{figure}[t]
   \centering
  \includegraphics[width=\columnwidth]{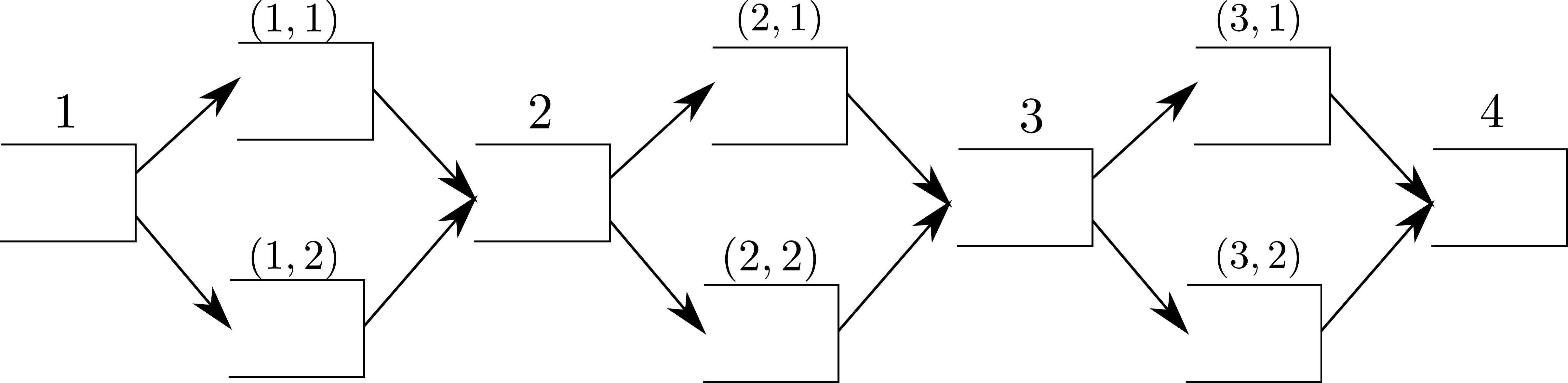}
  \caption{Queueing Network with Additional Precedence Constraints for DAG in Fig. \ref{fig:DAG}.}\label{fig:DAG_queue_constraint}
\end{figure}
\begin{example}
Consider a job specified by the DAG shown in Fig. \ref{fig:DAG} and a network consisting $J=2$ servers. We maintain the processing queues for each of $5$ possible stages of the job which are $\{1,2,3,4\}$, $\{2,3,4\}$, $\{2,4\}$, $\{3,4\}$ and $\{4\}$. Since task $1$ in stage $\{1,2,3,4\}$ is processable, we maintain communication queues $\{1,2,3,4\}_{(1,1)}$ and $\{1,2,3,4\}_{(1,2)}$ for server $1$ and server $2$ respectively. Similarly, we maintain communication queues $\{2,3,4\}_{(2,1)}$, $\{2,3,4\}_{(2,2)}$, $\{2,3,4\}_{(3,1)}$ and $\{2,3,4\}_{(3,2)}$ for stage $\{2,3,4\}$; communication queues $\{3,4\}_{(3,1)}$, $\{3,4\}_{(3,2)}$ for stage $\{3,4\}$; communication queues $\{2,4\}_{(2,1)}$, $\{2,4\}_{(2,2)}$ for stage $\{2,4\}$. The corresponding network of virtual queues is shown as Fig. \ref{fig:DAG_queue}. 
\end{example}
Now we describe the dynamics of the virtual queues in the network. When a new job $m$ comes to network, the job is sent to the processing queue corresponding to stage $\mathcal{S}_m=\mathcal{V}_m$ of job $m$. When server $j$ works on task $k$ in processing queue corresponding to subset $\mathcal{S}_m$, the result of the process is sent to the communication queue indexed by $\mathcal{S}_m$ and $(k,j)$ with rate $\mu_{(k,j)}$. When server $j$ broadcasts the result of processed task $k$ in the communication queue indexed by $\mathcal{S}_m$ and $(k,j)$, the result of the process is sent to the processing queue corresponding to subset $\mathcal{S}_m \backslash \{k\}$ with rate $\frac{b_j}{c_k}$. We call the action of processing task $k$ in processing queue corresponding to $\mathcal{S}_m$ as a \textit{processing activity}. Also, we call the action of broadcasting the output of processed task $k$ in communication queue as a \textit{communication activity}. We denote the collection of different processing activities and different communication activities in the network as $\mathcal{A}$ and $\mathcal{A}_c$ respectively. Let $A = |\mathcal{A}|$ and $A_c = |\mathcal{A}_c|$. Define the collection of processing activities that server $j$ can perform as $\mathcal{A}_j$, and the collection of communication activities that server $j$ can perform as $\mathcal{A}_{c,j}$. 

\begin{remark}
In general, the number of virtual queues corresponding to different stages of a job can grow exponentially with $K$ since each stage denotes a feasible subset of tasks. It can result in the increase of complexity of scheduling policies that try to maximize the throughput of the network. In terms of number of virtual queues, it is important to find a queueuing network with low complexity while resolving the problem of synchronization (see \cite{pedarsani2014scheduling} for more details) and guaranteeing the output of processed tasks to be sent to the same server.
\end{remark}
\begin{remark}
  To decrease the complexity of the queueing network, a queueing network with lower complexity can be formed by enforcing some additional constrains such that the DAG representing the job becomes a chain. As an example, if we force another constraint that task $3$ should proceed task $2$ in Fig. \ref{fig:DAG}, then the job becomes a chain of $4$ nodes with queueing network represented in Fig. \ref{fig:DAG_queue_constraint}. The queueing network of virtual queues for stages of the jobs has $K$ queues which largely decreases the complexity of scheduling policies for the DAG scheduling problem. 
\end{remark}
\subsection{Capacity Region}
Let $K'$ be the number of virtual queues for the network. For simplicity, we index the virtual queues by $k'$, $ k' \in [K']$. We define a drift matrix $D \in \mathbb{R}^{K' \times (A+A_c)}$ for $d_{(k',a)}$ where $d_{(k',a)}$ is the rate that virtual queue $k'$ changes if activity $a$ is performed. Define a length $K'$ arrival vector $\vec{e}(\vec{\lambda})$ such that $e_{k'}(\vec{\lambda}) = \lambda_m$ if virtual queue $k'$ corresponds to the first stage of jobs in which no tasks are yet processed, and $e_{k'}(\vec{\lambda})=0$ otherwise. Let $\vec{z} \in \mathbb{R}^{(A+A_c)}$ be the allocation vector of the activities $a \in \mathcal{A} \cup \mathcal{A}_c$. Similar to the capacity region introduced in the previous sections, we can introduce an optimization problem for the network that characterizes the capacity region. The optimization problem called \textit{broadcast planning problem (BPP)} is defined as follows:

\textbf{Broadcast Planning Problem (BPP):}
\begin{align}
    \text{Minimize} \quad  \eta \\
    \text{subject to} \quad & \vec{e}(\vec{\lambda})+D\vec{z} \leq \vec{0} \label{eq:incapacity}\\
    & \eta \geq \sum_{a \in \mathcal{A}_j} z_a, \ \forall \ j \in [J],\\
    & \eta \geq \sum_{a \in \mathcal{A}_{c,j}}z_a, \ \forall \ j \in [J],\\
    & \vec{z} \geq \vec{0}.
\end{align}
Based on BPP above, the capacity region of the network can be characterized by following proposition.
\begin{proposition}\label{optimal_BPP}
The capacity region $\Lambda^{''}$ of the virtual queueing network characterizes the set of all rate vectors $\vec{\lambda} \in \mathbb{R}^M_{+}$ for which the corresponding optimal solution $\eta^*$ to the broadcast planning problem (BPP) satisfies $\eta^* \leq 1$. In other words, \textit{capacity region} $\Lambda^{''}$ is characterized as follows 
\begin{align*}
   \Lambda^{''}  \triangleq \Bigg\{\vec{\lambda} \in \mathbb{R}^{M}_{+}: \exists \ \vec{z} \geq \vec{0} \ \text{such that} \ 1 \geq \sum_{a \in \mathcal{A}_j} z_a, \ \forall \  j, \\  1 \geq \sum_{a \in \mathcal{A}_{c,j}}z_a, \ \forall \ j, \ \text{and} \ \vec{e}(\vec{\lambda})+D\vec{z} \leq \vec{0}   \Bigg\}.
  \end{align*}
 \end{proposition}
 The proof of Proposition \ref{optimal_BPP} is similar to Proposition \ref{optimal_QNPP}. Due to the space limit, we only provide the proof sketches. Consider the virtual queueing network in the fluid limit. Suppose $\eta^*>1$. We assume that there exits a scheduling policy such that under that policy the virtual queueing network is weakly stable. Then, one can obtain a solution such that $\eta \leq 1$ which contradicts $\eta^*>1$. On the other hand, suppose $\eta^*\leq 1$. Similar to (\ref{eq:opt_rkj}) and (\ref{eq:opt_rkjc}), we can find a feasible allocation vector $\vec{z}$ such that the derivative of each queue's fluid level is not greater than $0$ which implies the network is weakly stable.
 \begin{remark}
Additional precedence constraints do not result in loss of throughput. Given a $\vec{\lambda} \in \Lambda^{''}$ with corresponding allocation vector $\vec{z}$, one can construct a feasible allocation vector $\vec{z}^{\,'}$ for the virtual queueing network based on additional precedence constraints: For each task $k$ and each server $j$, we choose the allocation of processing activity that task $k$ is processed by server $j$ to be the sum of all allocations which correspond to task $k$ and server $j$ in $\vec{z}$. For communication activity, it can be done in a similar argument. However, this serialization technique could increase the job latency (delay) as it de-parallelizes computation tasks. Also, the gap could be huge if the original DAG of job has a large number of tasks stemming from a common parent node. 
\end{remark}
\subsection{Throughput-Optimal Policy}
Now, we propose Max-Weight policy for the queueing network and show that it is throughput-optimal. Given virtual queue-lengths $Q^n_{k'}$ at time $n$, Max-Weight policy allocates a vector $\vec{z}$ that is
\begin{align}
&\arg \max_{\vec{z} \ \text{are feasible}} -(\vec{Q}^n)^{T} E[\Delta \vec{Q}^n|\mathcal{F}^n]\\
= & \arg \max_{\vec{z} \ \text{are feasible}} -(\vec{Q}^n)^{T}D\vec{z}
\end{align}
where $\vec{Q}^n $ is the vector of queue-lengths $Q^n_{k'}$ at time $n$. Next, we state the following theorem for the throughput-optimality of Max-Weight Policy.
 \begin{theorem}
 \label{optimal_broadcast}
 Max-Weight policy is throughput-optimal for the queueing network proposed in Subsection \ref{subsec:DAG_queue}.
 \end{theorem}
\begin{proof}
We consider the problem in the fluid limit. Define the amount of fluid in virtual queue $k'$ as $X_{k'}(t)$. The dynamics of the fluid are as follows:
\begin{equation}
    \vec{X}(t) = \vec{X}(0) + \vec{e}(\vec{\lambda})t + D\vec{T}(t),
\end{equation}
 where $\vec{X}(t)$ is the vector of queue-lengths $X_{k'}(t)$, $T_a(t)$ is the total time up to $t$ that activity $a$ is performed, and $\vec{T}(t)$ is the vector of total service times of different activities $T_a(t)$. By Max-weight policy, we have
 \begin{equation}
  \dot{\vec{T}}_{\text{Max-Weight}}(t) = \arg \min_{\vec{z} \text{ is feasible}} \vec{X}^T(t)D\vec{z}.
 \end{equation}
 Now, we take $V(t)=\frac{1}{2}\vec{X}^T\vec{X}$ as the Lyapunov function. The drift of $V(t)$ by using Max-Weight policy is
 \begin{align*}
     \dot{V}_{\text{Max-Weight}}(t) & = \vec{X}^T(t)(\vec{e}(\vec{\lambda})+D\vec{\dot{T}}_{\text{Max-Weight}}(t))\\
    & =\vec{X}^T(t)\vec{e}(\vec{\lambda})+\min_{\vec{z}\text{ is feasible}}\vec{X}^T(t)D\vec{z}\\
    & \leq \vec{X}^T(t)(\vec{e}(\vec{\lambda})+D\vec{z}^{\,*})
 \end{align*}
 where $\vec{z}^{\,*}$ is a feasible allocation vector. If $\vec{\lambda} \in \Lambda''$, then $\dot{V}(t) \leq 0$ which is directly from (\ref{eq:incapacity}). That is if $\vec{X}(0)=\vec{0}$, then $\vec{X}(t)=\vec{0}$ for all $t \geq 0$ which implies that the fluid model is weakly stable, i.e. the queueing network is rate stable \cite{dai1995positive}. 
 
 If $\vec{\lambda}$ is in the interior of capacity region $\Lambda''$, we have
 \begin{align}
    \dot{V}_{\text{Max-Weight}}(t) \leq \vec{X}^T(t)(\vec{e}(\vec{\lambda})+D\vec{z}^{\,*}) < \vec{0}, 
 \end{align}
 which proves that the fluid model is stable which implies the positive recurrence of the underlying Markov chain \cite{dai1995positive}.
\end{proof}

\section{Conclusion}
In this paper, we consider the problem of communication-aware dynamic scheduling of serial tasks for dispersed computing, motivated by significant communication costs in dispersed computing networks. We characterize the capacity region of the network and propose a novel network of virtual queues encoding the state of the network. Then, we propose a Max-Weight type scheduling policy, and show that the policy is throughput-optimal through a Lyapunov argument by considering the virtual queueing network in the fluid limit. Lastly, we extend our work to communication-aware DAG scheduling problem under a broadcast network where servers always broadcast the output of processed tasks to other servers. We propose a virtual queueing network encoding the state of network which guarantees the results of processed parents tasks are sent to the same server for processing child task, and show that the Max-Weight policy is throughput-optimal for the broadcast network. Some future directions are to characterize the delay properties of the proposed policy, develop robust scheduling policies that are oblivious to detailed system parameters such as service rates, and  develop low complexity and throughput-optimal policies for DAG scheduling. 

Beyond these directions, another future research direction is to consider communication-aware task scheduling when \textit{coded computing} is also allowed. 
Coded computing is a recent technique that enables optimal tradeoffs between computation load, communication load, and computation latency due to stragglers in distributed computing (see, e.g., ~\cite{li2018fundamental,li2017coding,lee2018speeding,dutta2016short,reisizadeh2019coded}). Therefore, designing joint task scheduling and coded computing in order to leverage tradeoffs between computation, communication, and latency could be an interesting problem (e.g., \cite{yang2019timely}).

\bibliographystyle{ieeetr}
\bibliography{references.bib}
\appendices
\section{Proof of Lemma \ref{equivalence}}\label{appendix:lemma1}
First, consider a vector $\vec{\lambda} \in \Lambda^{'}$. There exist feasible allocation vectors $\vec{p}$, $\vec{q}$, $\vec{u}$, $\vec{s}$ and $\vec{w}$ such that
\begin{align}
& \mu_{(k,j)}p_{(k,j)} = r_{(k,j)}, \ \forall \ j,\ \forall \ k. \\ 
& \frac{b_jq_{(k,j)}}{c_k} = r_{(k,j),c}=r_{(k,j)}(1-s_{k,j \rightarrow j}),\ \forall \ j, \ \forall \ k\in [K] \backslash \mathcal{H}.
\end{align}
Now, we focus on job $m$ specified by a chain and compute $\sum^J_{j=1}r_{(k,j)}$. If $k \in \mathcal{C}$, we have
\begin{align*}
    \sum^J_{j=1}r_{(k,j)} = \sum^J_{j=1}\lambda_{m(k)}u_{m \rightarrow j} = \lambda_{m(k)} \sum^J_{j=1}u_{m \rightarrow j}= \lambda_{m(k)}
\end{align*}
since $\sum^J_{j=1}u_{m \rightarrow j}=1$. Then, we can compute $\sum^J_{j=1} r_{(k+1,j)}$ as follows:
\begin{align}
\sum^J_{j=1}r_{(k+1,j)} & = \sum^J_{j=1} \sum^J_{l=1} r_{(k,l)} f_{k,l \rightarrow j} =\sum^J_{l=1} r_{(k,l)} \sum^J_{j=1} f_{k,l \rightarrow j} \nonumber\\
& = \sum^J_{l=1} r_{(k,l)}=\lambda_{m(k)}
\end{align}
since $\sum^J_{j=1} f_{k,l \rightarrow j}=1$. By induction, we have $\sum^J_{j=1}r_{(k,j)} = \lambda_{m(k)}$, $\forall k \in \mathcal{I}_m$. Then, we have
\begin{align}
    \sum^J_{j=1}r_{(k,j)} = \sum^J_{j=1}\mu_{(k,j)}p_{(k,j)},
\end{align}
which concludes $\nu_k(\vec{\lambda}) = \lambda_{m(k)} = \sum^J_{j=1}\mu_{(k,j)}p_{(k,j)}$ for all $k$. For $\forall k \in [K] \backslash \mathcal{H}$ and all $j$, we can write
\begin{align}
\frac{b_jq_{(k,j)}}{c_k} & = r_{(k,j),c} \\
& = r_{(k,j)}(1-s_{k,j \rightarrow j})\\
& = r_{(k,j)}-r_{(k,j)}f_{k,j \rightarrow j}\\
& \geq r_{(k,j)}-\sum^J_{l=1}r_{(k,l)}f_{k,l \rightarrow j}\\
& = r_{(k,j)} - r_{(k+1,j)}\\
& = \mu_{(k,j)}p_{(k,j)} - \mu_{(k+1,j)}p_{(k+1,j)}.
\end{align}
Thus, $\Lambda^{'} \subseteq \Lambda$.

Now, we consider a rate vector $\vec{\lambda} \in \Lambda$. There exist allocation vectors $\vec{p}$ and $\vec{q}$ such that $\nu_k(\vec{\lambda})=\sum^{J}_{j=1}\mu_{(k,j)}p_{(k,j)}$, $\forall \ k $; and $\frac{b_jq_{(k,j)}}{c_k} \geq \mu_{(k,j)}p_{(k,j)} - \mu_{(k+1,j)}p_{(k+1,j)}$, $\forall \ j$, $\forall \ k \in [K] \backslash \mathcal{H}$. For QNPP, for $\forall \ m \in [M]$, one can simply choose $u_{m \rightarrow j}$ as follows:
\begin{align}
    u_{m \rightarrow j} = \frac{\mu_{(k,j)}p_{(k,j)}}{\nu_k(\lambda)}, \ \forall \ j,
\end{align}
where $k$ is the root node of job $m$. For $k \in [K] \backslash \mathcal{H}$, we denote
\begin{equation}
\mathcal{D}_k = \{j:\mu_{(k,j)}p_{(k,j)}-\mu_{(k+1,j)}p_{(k+1,j)} <0,\ \forall \ j \}. 
\end{equation}
Then, for $\ k \in [K] \backslash \mathcal{H}$, we choose $s_{k,j \rightarrow j}$ as follows:
\begin{align}
    s_{k,j \rightarrow j} = 
    \begin{cases}
     1 & \text{if}\ j \in \mathcal{D}_k\\
     \frac{\mu_{(k+1,j)}p_{(k+1,j)}}{\mu_{(k,j)}p_{(k,j)}} & \text{if} \ j \notin \mathcal{D}_k\\
     \end{cases}
\end{align}
For $j \in \mathcal{D}_k$, we choose $w_{k,j \rightarrow l}$ to be any feasible value such that
\begin{align}
\sum_{l \in [J] \backslash\{j\}}w_{k,j \rightarrow l} =1.
\end{align}
For $j \notin \mathcal{D}_k$, we choose $w_{k,j \rightarrow l}$ as follows:
\begin{align}
   \label{eq:s_kjl}
    w_{k,j \rightarrow l} = 
    \begin{cases}
     \frac{\mu_{(k+1,l)}p_{(k+1,l)}-\mu_{(k,l)}p_{(k,l)}}{\sum_{l \in \mathcal{D}_k}\mu_{(k+1,l)}p_{(k+1,l)}-\mu_{(k,l)}p_{(k,l)}} & \text{if} \ l \in \mathcal{D}_k\\
     0 &  \text{if} \ l \notin \mathcal{D}_k. \\
     \end{cases}
\end{align}
One can easily check that $\vec{s}$ and $\vec{w}$ are feasible.
Based on the feasible vectors $\vec{u}$, $\vec{s}$ and $\vec{w}$ stated above, we can compute $r_{(k,j)}$. Let's focus on job $m$ and compute nominal rate $r_{(k,j)}$. If $k \in \mathcal{C}$, we have
\begin{align}
r_{(k,j)} = \nu_k(\lambda) u_{m \rightarrow j} = \mu_{(k,j)}p_{(k,j)}, \ \forall \ j. 
\end{align}
Then, we can compute $r_{(k+1,j)}$ in the following two cases:

\textbf{Case 1:} $j \in \mathcal{D}_k$: We compute $r_{(k+1,j)}$ as
\begin{flalign}
& r_{(k+1,j)} =  \sum^J_{l=1} r_{(k,l)}f_{k,l \rightarrow j} & \\
= & \mu_{(k,j)}p_{(k,j)}f_{k,j \rightarrow j}+\sum_{l \in [J] \backslash \{j\}} \mu_{(k,l)}p_{(k,l)}f_{k,l \rightarrow j} &\\
= & \mu_{(k,j)}p_{(k,j)}s_{k,j \rightarrow j}+\sum_{l \in [J] \backslash \{j\}} \mu_{(k,l)}p_{(k,l)}(1-s_{k,l \rightarrow l})w_{k,l \rightarrow j}\\
= & \mu_{(k,j)}p_{(k,j)}+ \sum_{l\notin \mathcal{D}_k}\mu_{(k,l)}p_{(k,l)}(1-\frac{\mu_{(k+1,l)}p_{(k+1,l)}}{\mu_{(k,l)}p_{(k,l)}}) \nonumber\\
& \times \frac{\mu_{(k+1,j)}p_{(k+1,j)}-\mu_{(k,j)}p_{(k,j)}}{\sum_{l \in \mathcal{D}_k}\mu_{(k+1,l)}p_{(k+1,l)}-\mu_{(k,l)}p_{(k,l)}}\\
= & \mu_{(k,j)}p_{(k,j)}+(\mu_{(k+1,j)}p_{(k+1,j)}-\mu_{(k,j)}p_{(k,j)}) \nonumber \\
& \times \frac{\sum_{l\notin \mathcal{D}_k}\mu_{(k,l)}p_{(k,l)} - \mu_{(k+1,l)}p_{(k+1,l)}}{\sum_{l \in \mathcal{D}_k}\mu_{(k+1,l)}p_{(k+1,l)}-\mu_{(k,l)}p_{(k,l)}}\\
= & \mu_{(k,j)}p_{(k,j)}+(\mu_{(k+1,j)}p_{(k+1,j)}-\mu_{(k,j)}p_{(k,j)})\\
= & \mu_{(k+1,j)}p_{(k+1,j)}
\end{flalign}
using the fact $\nu_k(\lambda) = \nu_{k+1}(\lambda)$, i.e. $\sum^{J}_{j=1}\mu_{(k,j)}p_{(k,j)} = \sum^{J}_{j=1}\mu_{(k+1,j)}p_{(k+1,j)}$.

\textbf{Case 2:} $j \notin \mathcal{D}_k$: We compute $r_{(k+1,j)}$ as
\begin{flalign}
& r_{(k+1,j)} = \sum^J_{l=1} r_{(k,l)}f_{k,l \rightarrow j} \\
= & \mu_{(k,j)}p_{(k,j)}f_{k,j \rightarrow j}+\sum_{l \in [J] \backslash \{j\}} \mu_{(k,l)}p_{(k,l)}f_{k,l \rightarrow j}\\
= & \mu_{(k,j)}p_{(k,j)}s_{k,j \rightarrow j} + \sum_{l \in [J] \backslash \{j\}} \mu_{(k,l)}p_{(k,l)}(1-s_{k,l \rightarrow l})w_{k,l \rightarrow j}\\
= & \mu_{(k+1,j)}p_{(k+1,j)}
\end{flalign}
since $s_{k,l \rightarrow l} =1$ for $l \in \mathcal{D}_k$ and $w_{k,l \rightarrow j}=0$ for $l,j \notin \mathcal{D}_k$ and $l \neq j$. Similarly, we can obtain $r_{(k,j)}=\mu_{(k,j)}p_{(k,j)}$ for $\forall \ k \in \mathcal{I}_m$. Now, for $k \in [K] \backslash \mathcal{H}$, we can compute $r_{(k,j),c}$. There are two cases as follows:

\textbf{Case 1:} $j \in \mathcal{D}_k$: We compute $r_{(k,j),c}$ as
\begin{align}
r_{(k,j),c} = & r_{(k,j)}(1-s_{k,j \rightarrow j})\\
= & \mu_{(k,j)}p_{(k,j)}(1-s_{k,j \rightarrow j}) = 0
\end{align}
since $s_{k,j \rightarrow j}=1$ for $j \in \mathcal{D}_k$. Therefore, $\frac{b_jq_{(k,j)}}{c_k} \geq 0 = r_{(k,j),c}$.

\textbf{Case 2:} $j \notin \mathcal{D}_k$: We compute $r_{(k,j),c}$ as
\begin{align}
r_{(k,j),c} = & r_{(k,j)}(1-s_{k,j \rightarrow j})\\
= & \mu_{(k,j)}p_{(k,j)} (1 - \frac{\mu_{(k+1,j)}p_{(k+1,j)}}{\mu_{(k,j)}p_{(k,j)}})\\
= & \mu_{(k,j)}p_{(k,j)} - \mu_{(k+1,j)}p_{(k+1,j)}
\end{align}
since $s_{k,j \rightarrow j}=\frac{\mu_{(k+1,j)}p_{(k+1,j)}}{\mu_{(k,j)}p_{(k,j)}}$ for $j \notin \mathcal{D}_k$. Then, we have $\frac{b_jq_{(k,j)}}{c_k} \geq \mu_{(k,j)}p_{(k,j)} - \mu_{(k+1,j)}p_{(k+1,j)} =r_{(k,j),c}$. Thus, $\Lambda \subseteq \Lambda^{'}$ which completes the proof.

\end{document}